\documentclass[12pt,a4paper]{article}
\usepackage{fullpage}
\usepackage{rotating}
\usepackage{lineno}

\usepackage{graphicx}
\usepackage{epsfig}
\usepackage{amsmath,amsfonts,amssymb,amsthm}
\usepackage{color}
\usepackage{hyperref}
\usepackage{natbib}
\usepackage[blocks]{authblk}
\usepackage{hyperref}
\hypersetup{colorlinks=true, citecolor=green}
\usepackage{setspace,enumitem}
\usepackage{color}
\usepackage{fancyvrb}
\usepackage{array}
\newcolumntype{C}[1]{>{\centering\let\newline\\\arraybackslash\hspace{0pt}}m{#1}}
\definecolor{synred}{RGB}{153,0,0}

\newcommand{\BEAS}{\begin{eqnarray*}}
\newcommand{\EEAS}{\end{eqnarray*}}
\newcommand{\BEA}{\begin{eqnarray}}
\newcommand{\EEA}{\end{eqnarray}}
\newcommand{\BIT}{\begin{itemize}}
\newcommand{\EIT}{\end{itemize}}
\newcommand{\BNUM}{\begin{enumerate}}
\newcommand{\ENUM}{\end{enumerate}}

\newcommand{\MC}{\mathcal}
\newcommand{\MB}{\mathbb}

\newcommand{\NN}{\nonumber}

\newcommand{\argmin}{\mathop{\rm argmin}}

\newcommand{\minimize}{\mathop{\rm minimize}}

\newcommand{\tr}{\mathrm{tr}}
\newcommand{\var}{\mathbb{V}{\rm ar}}

\newcommand{\sign}{\mathrm{sign}}

\newtheorem{theorem}{Theorem}

\newtheorem{lemma}{Lemma}

\title{Fixed support positive-definite modification of covariance matrix estimators
via linear shrinkage}
\date{}
\author{Young-Geun Choi}
\affil{Public Health Sciences Division,\\Fred Hutchinson Cancer Research Center, Seattle, WA, USA}
\author{Johan Lim}
\affil{Department of Statistics, Seoul National University, Seoul, Korea}
\author{Anindya Roy and Junyong Park}
\affil{Department of Mathematics and Statistics, University of Maryland, Baltimore County, MD, USA}
\begin{document}

\maketitle
\baselineskip 24pt
\begin{abstract}
\noindent 
In this work, we study the positive definiteness (PDness) problem in
covariance matrix estimation. For high dimensional data, 
many regularized estimators are proposed under structural assumptions
on the true covariance matrix including sparsity. They are shown to be asymptotically
consistent and rate-optimal in estimating the true
covariance matrix and its structure. However, many of them do not
take into account the PDness of the estimator and produce a non-PD estimate.
To achieve the PDness, researchers 
consider additional regularizations (or constraints) on
eigenvalues, which make both the asymptotic analysis and
computation much harder. In this paper, we propose a simple
modification of the regularized covariance matrix
estimator to make it PD while preserving the support.  
We revisit the idea of linear shrinkage and propose to take a convex combination between
the first-stage estimator (the regularized covariance matrix without
PDness) and a given form of diagonal matrix.
The proposed modification, which we denote as FSPD (Fixed Support and Positive 
Definiteness) estimator, is shown to preserve the asymptotic properties of the first-stage estimator,
if the shrinkage parameters are carefully selected. It has a closed form expression
and its computation is optimization-free, unlike existing PD sparse estimators.
In addition, the FSPD is generic in the sense that it can be applied to any non-PD matrix including the precision matrix. 
The FSPD estimator is numerically
compared with other sparse PD estimators to understand its finite sample properties
as well as its computational gain. It is also applied to two multivariate procedures
relying on the covariance matrix estimator --- the linear minimax classification problem 
and the Markowitz portfolio optimization problem --- and
is shown to substantially improve the performance of both procedures. 

\end{abstract}

\section{Introduction}\label{sec:introduction}
Covariance matrix and its consistent estimation are involved in many multivariate statistical procedures, where the sample
covariance matrix is popularly used. In recent, high dimensional data are prevalent everywhere for which the sample covariance
matrix is known to be inconsistent \citep{Marcenko1967}. To resolve the difficulty from high dimensionality, 
regularized procedures
or estimators are proposed under various structural assumptions on the true matrix.  For instance, if the true covariance matrix is
assumed to be sparse or banded, one thresholds the elements of the sample covariance matrix to satisfy the assumptions
\citep{Bickel2008a,Bickel2008,Cai2010,Cai2011a,Cai2012c,Cai2012f,Rothman2009}
or penalizes the Gaussian likelihood \citep{Bien2011,Lam2009}. The asymptotic theory of
the regularized estimators are well understood
and, particularly, they are shown to be consistent in estimating the true covariance matrix and its support (the positions of its non-zero elements).
The main interest of this paper is positive definiteness (PDness) of covariance matrix estimator.
The PDness is an essential property for the validity of many multivariate statistical procedures.
However, the regularized covariance matrix estimators recently studied are often not PD in finite sample. This is because they more focus
on the given structural assumptions and do not impose the PDness on their estimators. For example, the banding or
thresholding
 method \citep{Bickel2008,Rothman2009} regularizes the sample covariance matrix in an elementwise manner and provides an
 explicit form of the estimator that satisfies the given assumptions. Nonetheless, 
 the eigenstructure of the resulting estimator is 
 completely unknown 
 and, without doubt, the resulting covariance matrix estimate is not necessarily PD.
 A few efforts are made to find an estimator which attains both the sparsity and PDness 
by \emph{incorporating} them in a single optimization problem \citep{Bien2011,Lam2009,Liu2014,Rothman2012,Xue2012}.
In particular,
 the works by \citet{Rothman2012}, \citet{Xue2012}, and \citet{Liu2014} understand the soft thresholding of the sample
 covariance (correlation) matrix as a convex minimization problem and
add a convex penalty (or constraint) to the problem in order to guarantee the PDness of solution. However, 
we remark that each of these incorporating approaches has to be customized to a certain regularization technique
(e.g. soft thresholding estimator)  and also requires us to solve a large-scale optimization problem.

Instead of simultaneously handling PDness and regularization, we propose a \emph{separated} update of a given 
covariance matrix estimator. 
Our motivation is that the regularized estimators in the literature are already proven to be ``good'' in terms of
consistency or rate-optimality in estimating their true counterparts. Thus, we aim to
minorly modify them to retain the same asymptotic properties as well as be PD with the same support. To be specific,
denote by $\widehat{\bf \Sigma}$ a given covariance matrix estimator.
We consider a distance minimization problem
\begin{equation} \label{eqn:ideal-opt}
\minimize_{\widehat{\bf \Sigma}^*}\,
\bigg\{ \Big\| \widehat{\bf \Sigma}^* - \widehat{\bf \Sigma} \Big\| \,:\,
	\gamma_{1}(\widehat{\bf \Sigma}^*) \geq \epsilon
	,~ {\rm supp}(\widehat{\bf \Sigma}^*) = {\rm supp}(\widehat{\bf \Sigma})
	, ~ \widehat{\bf \Sigma}^* = (\widehat{\bf \Sigma}^*)^{\top}
\bigg\}
\end{equation}
where $\epsilon >0$ is a pre-determined small constant and $\gamma_{1}(\widehat{\bf \Sigma}^*)$ denotes the smallest
eigenvalue of
$\widehat{\bf \Sigma}^*$.
In solving (\ref{eqn:ideal-opt}), to make the modification simple, we further restrict the class of $\widehat{\bf \Sigma}^*$ to
a family of linear
shrinkage of $\widehat{\bf \Sigma}$ to the identity
matrix that is
\begin{equation} \label{eqn:lspd}
\widehat{\bf \Sigma}^* ~ \in \left\{ {\bf \Phi}_{\mu, \alpha}\big(\widehat{\bf \Sigma}\big) \equiv \alpha
\widehat{\bf \Sigma} + (1 - \alpha) \mu {\bf I} ~:~ \alpha \in [0,1], ~\mu \in \mathbb{R} \right\}.
\end{equation}
The primary motivation for considering \eqref{eqn:lspd} is that shrinking $\widehat{\bf \Sigma}$ linearly to the identity enables us to
handle the eigenvalues of $\widehat{\bf \Sigma}^*$ easily while preserving the support of
$\widehat{\bf \Sigma}$.
We will reserve  the term \emph{fixed support PD} (FSPD) estimator to describe any  estimator of the form \eqref{eqn:lspd} that solves the minimization problem \eqref{eqn:ideal-opt}, and refer to the process of
 modification as an FSPD procedure.
The proposed FSPD estimator/procedure has several interesting and important  properties.
 First, 
the calculation of ${\bf \Phi}_{\mu, \alpha}\big(\widehat{\bf \Sigma}\big)$ is optimization-free, since the choice of $\mu$ and $\alpha$ can be explicitly expressible with the smallest and largest eigenvalues of the initial estimator  $\widehat{\bf \Sigma}$. Second, for suitable choices of $\mu$ and $\alpha$ in (\ref{eqn:lspd}), the  estimators ${\bf \Phi}_{\mu, \alpha}\big(\widehat{\bf \Sigma}\big)$  and $\widehat{\bf \Sigma}$ have the same  rate of convergence to the true covariance matrix $\bf \Sigma$ under some conditions. Third, owing to the separating nature, the FSPD procedure is equally applicable to any (possibly)  non-PD estimator of covariance matrix as well as precision matrix (the inverse of covariance matrix).

The rest of the paper is organized as follows.
In Section \ref{sec:previous},
we illustrate some simulated examples of estimators that have non-PD outcomes.
Recent state-of-the-art sparse covariance estimators which guarantee the PDness are also
briefly reviewed. Our main results are presented in Section \ref{sec:LSPD}.
The FSPD procedure is developed
by solving the restricted distance minimization presented in (\ref{eqn:ideal-opt}) and (\ref{eqn:lspd}).
We not only derive statistical convergence rate
of the resulting FSPD estimator, but also discuss some implementation issues for its practical use.
In Section \ref{sec:simulation}, we numerically show that FSPD estimator has
 comparable risks with recent PD sparse covariance matrix estimators introduced in Section \ref{sec:previous}, whereas ours
 are computationally much simpler and faster. In Section \ref{sec:plug-in}, we illustrate the usefulness of FSPD-updated regularized covariance estimators in two statistical procedures: the linear minimax classification problem
  \citep{Lanckriet2002} and Markowitz portfolio optimization with no short-sale \citep{Jagannathan2003}.
Since the FSPD procedure is applicable to any covariance matrix estimators including precision matrix
estimators, we briefly discuss this extendability in Section \ref{sec:prec}.
Finally, Section \ref{sec:concluding} is for concluding remarks.
%%%}

{\bf Notations:} We assume all covariance matrices are of size $p \times p$.  Let ${\bf \Sigma}$,
$\bf S$ and $\widehat{\bf \Sigma}$ be the true covariance matrix, the sample covariance matrix, and a generic covariance 
matrix estimator, respectively. For a symmetric matrix ${\bf A}$, the ordered  eigenvalue will be denoted by $\gamma_1({\bf A}) \leq \cdots \leq \gamma_p({\bf A})$. In particular, we abbreviate $\gamma_i({\bf \Sigma})$ to $\gamma_i$ and $\gamma_i(\widehat{\bf \Sigma})$ to $\widehat{\gamma}_i$. The Frobenius norm of ${\bf A}$ is
 defined with scaling by $\|{\bf A}\|_{\rm F} := \sqrt{\tr({\bf A}^{\top}{\bf A})/p}$. The spectral norm
 of ${\bf A}$ is $\|{\bf A}\|_2 := \sqrt{\gamma_{p}({\bf A}^{\top}{\bf A})} \equiv \max_i |\gamma_{i}({\bf A})|$. The
 norm without subscription, $\|\cdot\|$, will be used in the case both the spectral and Frobenius norm are applicable.

\section{Covariance regularization and PDness}\label{sec:previous}

\subsection{Non-PDness of regularized covariance matrix estimators}

In this section, we briefly review two most common regularized covariance matrix estimators and
discuss their PDness.

\subsubsection{Thresholding estimators}
The thresholded sample covariance matrix (simply \emph{thresholding
 estimator}) regards small elements in the sample covariance matrix as noise and set them to zero. For a fixed
 $\lambda$, the estimator is defined by
\begin{equation}\label{estimator:thr}
\widehat{\bf \Sigma}_{\lambda}^{\rm Thr} := \Big[ T_{\lambda}(s_{ij}), 1 \le i,j \le p \Big],
\end{equation}
where $T_{\lambda}(\cdot) : \MB{R} \rightarrow \MB{R}$ is a thresholding function \citep{Rothman2009}.
Some examples of the thresholding function  are
(i) (hard thresholding) $T_{\lambda}(s) = I(|s| \geq \lambda) \cdot s$,
 (ii) (soft thresholding) $T_{\lambda}(s) = \sign(s) \cdot (|s| - \lambda)_+$, and (iii) (SCAD thresholding)
$T_{\lambda}(s) = I(|s| < 2\lambda)\cdot \sign(s) \cdot (|s| - \lambda)_+$  $+$
$I(2\lambda < |s| \leq a\lambda) \cdot \{(a - 1)s - \sign(s) a \lambda\} / (a - 2)$ $+$
$I(|s| > a\lambda)\cdot s$ with $a>2$.
One may threshold only the off-diagonal elements of $\bf S$, in which case the estimator remains asymptotically
equivalent to (\ref{estimator:thr}).
 The universal threshold $\lambda$  in (\ref{estimator:thr}) can be adapted to each element;
 \cite{Cai2011b} suggests an \emph{adaptive thresholding estimator} that thresholds each $s_{ij}$ with an
 element-adaptive threshold $\lambda_{ij}$. In section \ref{sec:plug-in}, we will use a adaptive thresholding
 estimator with the soft thresholding function.
 
As we point out earlier, the elementwise manipulation in the thresholding estimators
does not retain its PDness.
The left panel of Figure \ref{figure:thresbanding} plots the minimum eigenvalues of thresholding
 estimators with various thresholding functions and tuning parameters, when the dataset of $n=100$ and $p=400$ are sampled from the multivariate $t$-distribution with $5$ degrees of freedom and true covariance matrix ${\bf M}_1$ defined in Section \ref{sec:simulation}.
As shown in
the figure, $\widehat{\bf \Sigma}_{\lambda}^{\rm Thr}$ is not PD if $\lambda$ is moderately small or selected via the five-fold cross-validation (CV) explained in \cite{Bickel2008a}.

\subsubsection{Banding estimators}
The banded covariance matrix arises when the variables of data are ordered and serially dependent as in the data of time series, climatology, or spectroscopy. \cite{Bickel2008} proposes a class of \emph{banding estimators}
\[
\widehat{\bf \Sigma}_h^{\rm Band} = \Big[ s_{ij} \cdot w_{|i-j|} \Big],
\]
where banding weight $w_m$ ($m = 0, 1, \ldots, p-1$) is proposed by $w_m = I( m \leq h)$ in the paper for a fixed
 bandwidth $h$. Roughly speaking, $\widehat{\bf \Sigma}_h^{\rm Band}$ discards the sample
covariance $s_{ij}$ if corresponding indices $i$ and $j$ are distant.  \cite{Cai2010}
considers a \emph{tapering estimator} which smooths the banding weight in $\widehat{\bf \Sigma}_h^{\rm Band}$ as
\[
w_m = \begin{cases} 1, & {\rm when~} m \leq h/2 \\
	2 - 2m/h, & {\rm when~} h/2 < m \leq h \\
	0, & {\rm otherwise}.
	\end{cases}
\]
We could see that both banding estimators are again from elementwise operations on
the sample covariance matrix and, as in the thresholding estimators, do not
retain the PDness of the sample estimator.
The right panel of Figure \ref{figure:thresbanding} tells that this is indeed. The
panel is based on the same simulated data set for the left panel and
shows that the two estimators result in non-PD estimates regardless
of the size of bandwidth.

\begin{figure}[htb!]
\centering
\begin{minipage}{.45\textwidth}
  \centering
  \includegraphics[width=1.0\linewidth]{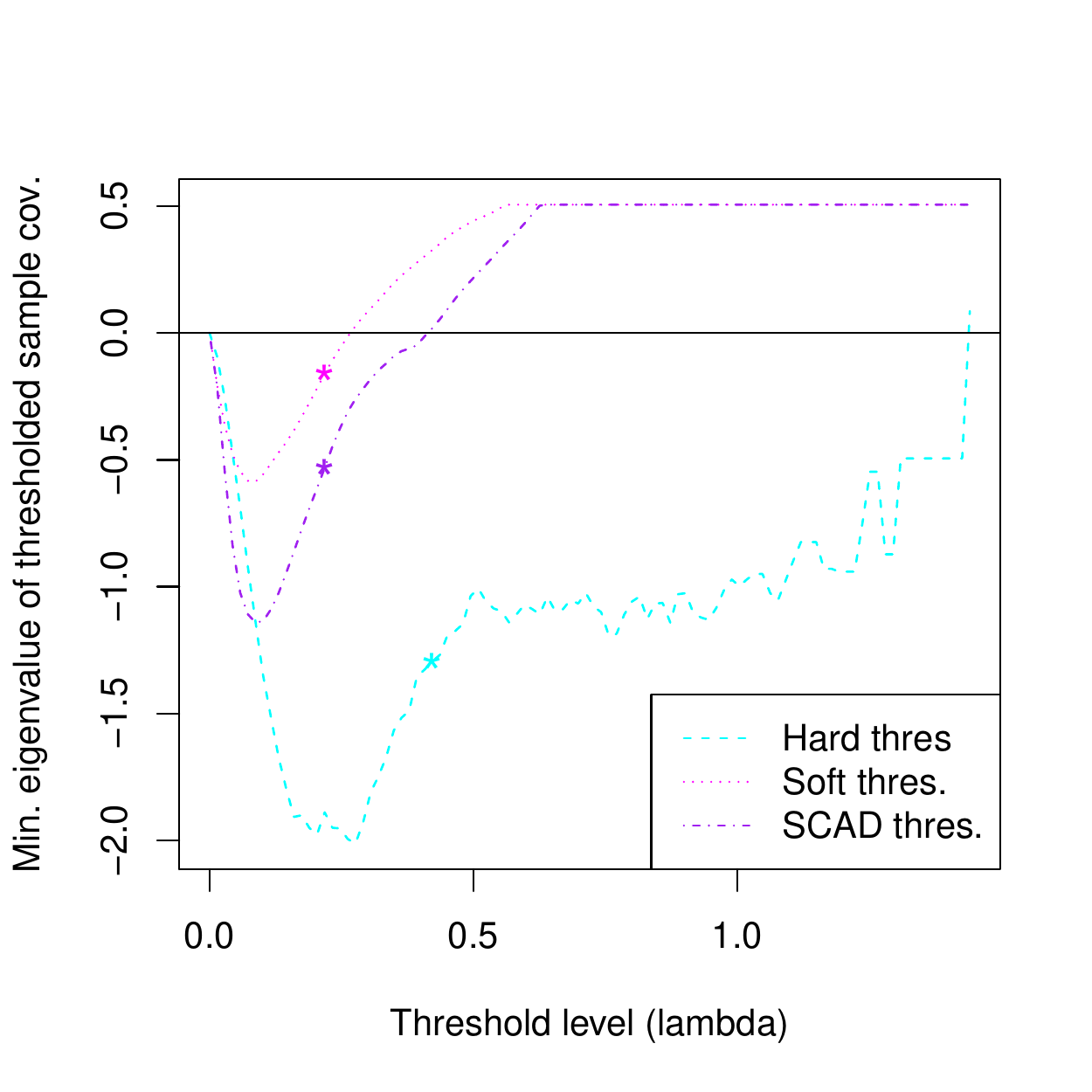}
 %\caption{A subfigure}
 % \label{fig:powerCSlamfix}
\end{minipage}%
\begin{minipage}{.45\textwidth}
  \centering
 \includegraphics[width=1.0\linewidth]{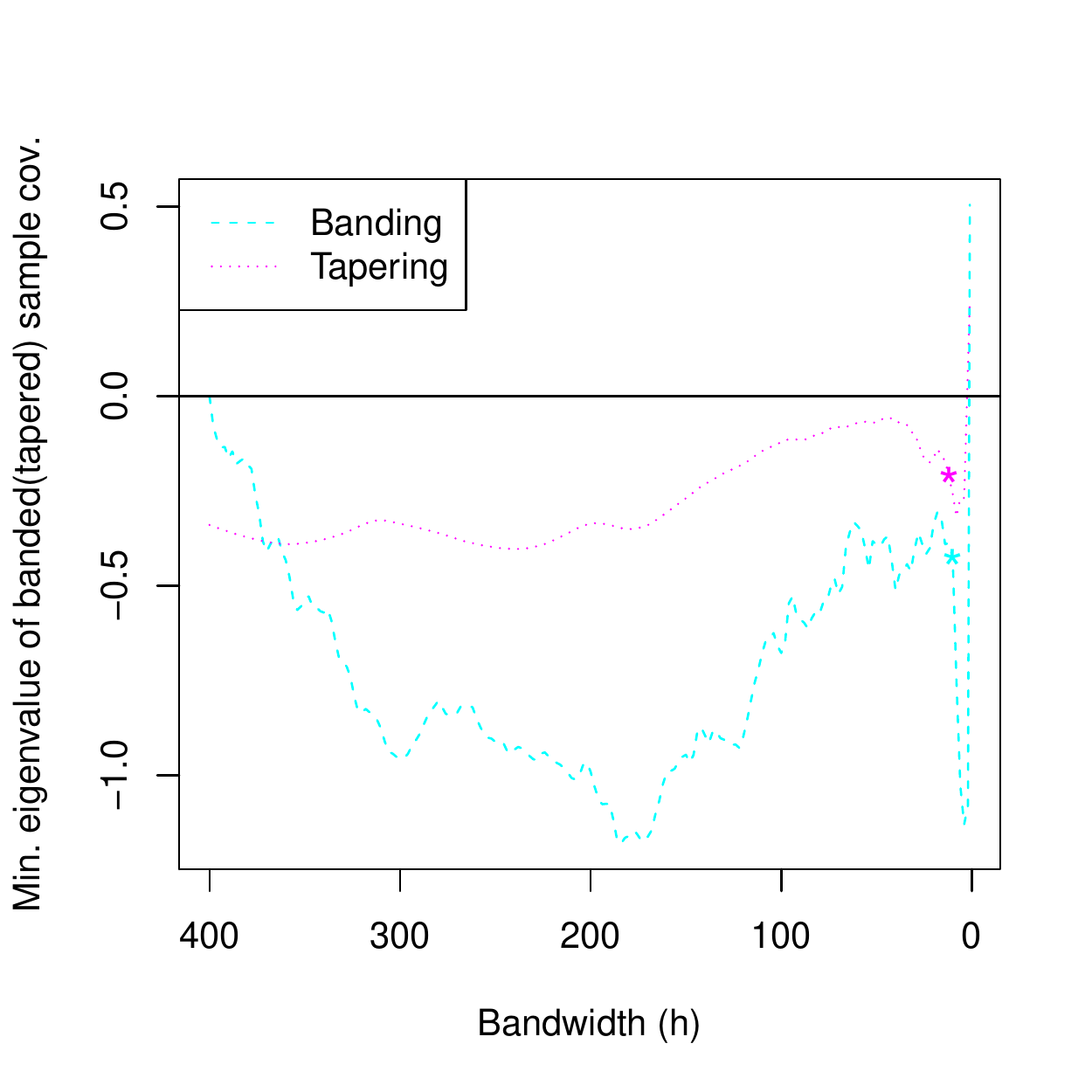}
 %\caption{A subfigure}
 % \label{fig:powerCSgammafix}
\end{minipage}
\caption{Miminum eigenvalues of regularized covariance matrix estimators
 ($n=100$, $p=400$). Left for thresholding estimators with threshold ($\lambda$) and
  right for banding estimators with varying bandwidth ($h$). The star ($*$)-marked point is
   the optimal threshold (bandwidth) selected by a five-fold cross-validation (CV). Note that all the
   CV-selected estimates are not PD.}
\label{figure:thresbanding}
\end{figure}

\subsection{Adding PDness to structural (sparse) regularization}\label{subsec:optPD}

The PDness of regularized covariance matrix estimator recently get attention of researchers and a limited number of works are done in the literature; see \citep{Rothman2012,Xue2012,Liu2014}. The three approaches, namely PD sparse estimators,
 are all based on the fact  that
the soft thresholding estimator (the thresholding estimator equipped with the soft thresholding function) can be obtained by minimizing the following convex
function:
\begin{equation}\label{estimator:soft}
\widehat{\bf \Sigma}^{\rm Soft}({\lambda}) = \Big[ (|s_{ij}| - \lambda)_+ \cdot \sign(s_{ij}) \Big] = \argmin_{{\bf \Sigma}} \,
	\|{\bf \Sigma} - {\bf S}\|_{\rm F}^2 + \lambda \sum_{i \leq j} |\sigma_{ij}|.
\end{equation}
In each work, a constraint or penalty function is added to (\ref{estimator:soft}) to encourage the solution to be PD. For example, \cite{Rothman2012} considers to add a log-determinant penalty:
\begin{equation}\label{estimator:logdet}
\widehat{\bf \Sigma}^{\rm logdet}(\lambda) := \argmin_{\bf \Sigma} \,
	\|{\bf \Sigma} - {\bf S}\|_{\rm F}^2 + \tau \log \det \left( {\bf \Sigma} \right) + \lambda \sum_{i < j} |\sigma_{ij}|,
\end{equation}
where $\tau > 0$ is fixed a small value.
The additional term behaves a convex barrier that naturally ensures the PDness of the
solution and preserves the convexity the objective function. In the paper, the author solves the normal equations
with respect to ${\bf \Sigma}$, where column vectors of the current estimate ${\bf \Sigma}$ are alternatingly updated by solving $(p-1)$-variate lasso regressions. We note that, even each lasso regression
can be calculated fast, repeating it for every column leads to $O(p^3)$ flop computations for the entire matrix to be updated.
On the other hand, \cite{Xue2012} proposes to solve
\begin{equation}\label{estimator:EigCon}
\widehat{\bf \Sigma}^{\rm EigCon}(\lambda) := \argmin_{{\bf \Sigma} \,:\, {\bf \Sigma} \succeq \epsilon{\bf I}} \,
	\|{\bf \Sigma} - {\bf S}\|_{\rm F}^2 + \lambda \sum_{i < j} |\sigma_{ij}|.
\end{equation}
where ${\bf \Sigma} \succeq \epsilon{\bf I}$ means that ${\bf \Sigma} - \epsilon{\bf I}$ is positive semidefinite.
Since the additional constraint $\{ {\bf \Sigma} \,:\, {\bf \Sigma} \succeq \epsilon{\bf I} \}$ is a convex set, (\ref{estimator:EigCon}) still has a global optimum.
The objective function is optimized via an alternating direction method of multipliers (ADMM) algorithm \citep{Boyd2010}, in which each iteration contains the eigenvalue decomposition and thresholding of a $p$ by $p$ matrix. The computational cost of the algorithm is, mainly due to eigenvalue decomposition, $O(p^3)$ flops per iteration which may be as demanding as Rothman's method. Similarly \cite{Liu2014} also considers (\ref{estimator:EigCon}) in which ${\bf S}$ is replaced by the sample correlation matrix.

We conduct a small simulation study on the above PD-sparse covariance matrix estimators to understand their empirical risk and computational hardness.
 Table  \ref{table:timecomp0} lists empirical risks and computation times of $\widehat{\bf \Sigma}^{\rm logdet}$, $\widehat{\bf \Sigma}^{\rm EigCon}$, and $\widehat{\bf \Sigma}^{\rm Soft}$, for a fixed tuning parameter, based on 100 replications from the same simulation settings as that used in  Figure \ref{figure:thresbanding}.
Both of the PD sparse covariance estimators show comparable performance to that of the soft thresholding estimator. %, which is shown to be rate-optimal \citep{Cai2012f}, in view of empirical risks.
However, the computation  times of PD sparse estimators
 are substantially increased.
According to the table, when the five-fold cross validations are done for tuning parameter selection with 100 candidates of $\lambda$,
it takes more than 2,500 seconds to get either $\widehat{\bf \Sigma}^{\rm logdet}$ or $\widehat{\bf \Sigma}^{\rm EigCon}$,
while $\widehat{\bf \Sigma}^{\rm Soft}$ only costs 6 seconds.
%%%%%%%%%%%%%%%%%%%%%% TIMING, NUMERIC %%%%%%%%%%%%%%%%%%%%%%
%%%%%%%%%%%%%%%%%%%%%%%%%%%%%%%%%%%%%%%%%%%%%%%%%%%%%%%%%%%%%
%%%%%%%%%%%%%%%%%%%%%%%%%%%%%%%%%%%%%%%%%%%%%%%%%%%%%%%%%%%%%
\begin{table}[h]{\small
\begin{center}
\begin{tabular}{c|ccccc}
\hline
	  & Matrix $l_1$ & Matrix $l_2$ & Frobenius & \#(PD) & Comp. Time (sec)  \\
\hline
Soft thresholding &  19.34  & 7.47   & 25.94  & 0/100 & 0.01  \\
Eigen. constraint  & 18.92  & 7.45  & 25.92  & 100/100 & 4.93 \\
log-det barrier  & 18.90  & 7.45  & 25.99  & 100/100  & 2.33 \\
\hline				
\end{tabular}
\caption{Averaged empirical risks, frequency of positive definite estimates, and computing time for one fixed tuning parameter under 100 replications,
measured on Intel Core i7-2600 3.4GHz CPU and 16GB RAM. The data is generated under
$n=100$, $p=400$, MV-$t$ distribution, and ${\bf \Sigma}= {\bf M}_1$ which is introduced in Section \ref{sec:simulation}.
Computational convergence criteria is set as the relative error be smaller than $10^{-7}$.
$\epsilon$ and $\tau$ are set as $10^{-2}$ for both the eigenvalue constraint method and the log-determinant barrier method.}
\label{table:timecomp0}
\end{center}
}
\end{table}
%%%%%%%%%%%%%%%%%%%%%%%%%%%%%%%%%%%%%%%%%%%%%%%%%%%%%%%%%%%%%
%%%%%%%%%%%%%%%%%%%%%%%%%%%%%%%%%%%%%%%%%%%%%%%%%%%%%%%%%%%%%
%%%%%%%%%%%%%%%%%%%%%%%%%%%%%%%%%%%%%%%%%%%%%%%%%%%%%%%%%%%%%

We finally remark that the above optimization-based PD sparse (regularized) estimators
are only applicable to a specific type of regularized estimator, that is expressible as the 
minimum of a convex function.
Note that not all regularized estimators have the convex expression; for example, the hard-thresholding 
estimator or SCAD estimator does not have the convex expression. In addition, even though an initial
regularized estimator can be written as the solution to a convex optimization problem, the modified objective 
function with the PD constraint (or penalty) strongly depends on the initial convex expression. Thus, the corresponding
PD regularized 
estimator should re-defined; the algorithm to solve the modified optimization problem should re-developed; 
and the statistical property of the resulting estimator should also be re-investigated, as a new.

\section{The linear shrinkage for fixed support positive definiteness (FSPD)}\label{sec:LSPD}

We now explain how the FSPD procedure is applied to any covariance matrix estimator that possibly lacks PDness. Let $\widehat{\bf \Sigma}$ be a given initial estimator. Recall that we propose a class of linear shrinkage as a modification of $\widehat{\bf \Sigma}$,
\begin{equation}  \label{eqn:lspd2}
{\bf \Phi}_{\mu, \alpha}\big(\widehat{\bf \Sigma}\big) := \alpha \widehat{\bf \Sigma} + (1 - \alpha) \mu {\bf I},
\end{equation}
where $\alpha \in [0,1]$ and $\mu \in \mathbb{R}$.
Our aim is to minimize the distance between
 ${\bf \Phi}_{\mu, \alpha}\big(\widehat{\bf \Sigma}\big)$ and $\widehat{\bf \Sigma}$, while
keeping the minimum eigenvalue of ${\bf \Phi}_{\mu, \alpha}\big(\widehat{\bf \Sigma}\big)$ 
positive. We state this constrained minimization problem quantitatively in Section \ref{subsec:distance} and derive the minimum for $\alpha$ while fixing $\mu$ as a constant in Section \ref{subsec:alpha}. Subsequently, we determine the condition of $\mu$ that minimizes the distances for the spectral and Frobenius norms in Section \ref{subsec:mu}. The statistical convergence rate of ${\bf \Phi}_{\mu, \alpha}\big(\widehat{\bf \Sigma}\big)$ for carefully chosen $\mu$ and $\alpha$ are established in Section \ref{subsec:convrate}. Moreover, we discuss 
the fast computation of the FSPD procedure in Section
\ref{subsec:computation}.

\subsection{Distance minimization}\label{subsec:distance}

We set a small cut-point $\epsilon > 0$ to determine whether $\widehat{\bf \Sigma}$ is PD; we will modify $\widehat{\bf \Sigma}$ if $\widehat{\gamma}_{1} < \epsilon$, otherwise we do not need to update it.
Let $\widehat{\gamma}_{1} < \epsilon$. We solve the following minimization problem:

\begin{eqnarray}
\minimize_{\mu, \alpha \in \MB{R}} &&
	\left\| {\bf \Phi}_{\mu, \alpha}\big(\widehat{\bf \Sigma}\big)
	- \widehat{\bf \Sigma} \right\| \label{objective:1} \\
\mbox{subject to} &&
	\begin{array}{l}
  	\alpha \widehat{\gamma}_1 + (1 - \alpha) \mu \geq \epsilon \, ; \\
  	\alpha \in [0,1).	
	\end{array} \NN
\end{eqnarray}
In (\ref{objective:1}), the first constraint enforces the minimum eigenvalue of ${\bf \Phi}_{\mu, \alpha}\big(\widehat{\bf \Sigma}\big)$ to be at least $\epsilon$. The second constraint specifies the range of $\alpha$, intensity of linear shrinkage, where $\alpha=0$ corresponds to the complete shrinkage to $\mu {\bf I}$  and $\alpha=1$ denotes no shrinkage.
 By the assumption $\widehat{\gamma}_{1} < \epsilon$, the two constraints also imply that $\mu \geq \epsilon$
\[
\mu \,\geq\,
	\frac{1}{1-\alpha} \epsilon - \frac{\alpha}{1-\alpha} \widehat{\gamma}_{1}
    \,>\, \frac{1}{1-\alpha} \epsilon - \frac{\alpha}{1-\alpha} \epsilon
    \,=\, \epsilon
\]
for any $\alpha \in [0,1)$.

\subsection{The choice of $\alpha$}\label{subsec:alpha}

Let  $\mu \in [\epsilon, \infty)$ be fixed. Provided that $\widehat{\gamma}_{1} < \epsilon$, 
we have $\widehat{\gamma}_{1} < \epsilon \leq \mu$ and
\begin{equation}\label{ineqn:alpha}
 1-\alpha \,\geq\,
 	\frac{\epsilon - \widehat{\gamma}_{1}}{\mu - \widehat{\gamma}_{1}}
\end{equation}
from the constraints of (\ref{objective:1}). Observe that the objective function satisfies 
$\| {\bf \Phi}_{\mu, \alpha}\big(\widehat{\bf \Sigma}\big) - \widehat{\bf \Sigma} \|$ $=$ $(1 - \alpha) \| \mu {\bf I} - \widehat{\bf \Sigma} \|$ and then achieves its minimum when $(1-\alpha)$ touches the lower bound in (\ref{ineqn:alpha}).
Thus, we have the following lemma.
\begin{lemma}\label{lemma:alpha} Let $\widehat{\bf \Sigma}$ be given and assume $\epsilon > \widehat{\gamma}_{1}$. Then for any given $\mu \in [\epsilon, \infty)$, the problem (\ref{objective:1}), with respect to $\alpha$, is minimized at
\begin{equation}\label{eqn:alpha}
  \alpha^* := \alpha^*(\mu) \,=\,
 	\frac{\mu - \epsilon}{\mu - \widehat{\gamma}_{1}}.
\end{equation}
\end{lemma}

\subsection{The choice of $\mu$}\label{subsec:mu}
By substituting (\ref{eqn:alpha}) into (\ref{objective:1}), we have a reduced problem, which depends only on $\mu$,
\begin{equation}\label{objective:2}
\minimize_{\mu \,:\, \mu > \epsilon} \quad
\left\| {\bf \Phi}_{\mu, \alpha^*}\Big(\widehat{\bf \Sigma}\Big)
	- \widehat{\bf \Sigma} \right\| ~ = ~
  	\frac{\epsilon - \widehat{\gamma}_{1}}{\mu - \widehat{\gamma}_{1}}
  	 \left\|   \mu {\bf I} - \widehat{\bf \Sigma}  \right\|.
\end{equation}
The solution of (\ref{objective:2}) differs according to how the matrix norm $\|\cdot\|$ is defined. We
consider two most popular matrix norms in below, the spectral norm $\|\cdot\|_2$ and the (scaled) Frobenius norm $\|\cdot\|_{\rm F}$.
\begin{lemma}[Spectral norm]\label{lemma:muspect} If $\widehat{\bf \Sigma}$ is given and $\epsilon > \widehat{\gamma}_{1}$, then
\[
\left\| {\bf \Phi}_{\mu, \alpha^*}\Big(\widehat{\bf \Sigma}\Big) 
	- \widehat{\bf \Sigma} \right\|_2
\,\geq\, \epsilon - \widehat{\gamma}_{1} ~~\mbox{ for all $\mu \geq \epsilon$},
\]
and the minimum of \eqref{objective:2} is achieved for any  $ \mu \geq 
	\mu_{\rm S} := \max \left\{ \epsilon, \frac{\widehat{\gamma}_{p} + \widehat{\gamma}_{1}}{2} \right\}$.
\end{lemma}
\begin{proof}%[Proof of Lemma (\ref{lemma:muspect}).]
By the definition of the spectral norm,
\begin{eqnarray}
\frac{\epsilon - \widehat{\gamma}_{1}}{\mu - \widehat{\gamma}_{1}} \left\|   \mu {\bf I} - \widehat{\bf \Sigma}  \right\|_2
&=&
\frac{\epsilon - \widehat{\gamma}_{1}}{\mu - \widehat{\gamma}_{1}} \max_{i} \left| \mu - \widehat{\gamma}_i \right|
\NN \\
&=&
\frac{\epsilon - \widehat{\gamma}_{1}}{\mu - \widehat{\gamma}_{1}} 
\max \left\{ | \mu - \widehat{\gamma}_{p} |, | \mu - \widehat{\gamma}_{1} | \right\}.
\NN
\end{eqnarray}
Consider $\psi_2(t) := \max \left\{ | t - a |, | t - b | \right\} / (t - a)$ with $a < b$, $t \in (a, \infty)$. One can easily verify that $\psi_2(t) \equiv 1$
 when $\mu \geq (a + b)/2$, and
$\psi_2(t) > 1$ when $a < t < (a + b)/2$. Now substitute $t \leftarrow \mu$, $a \leftarrow \widehat{\gamma}_{1}$, and $b \leftarrow \widehat{\gamma}_{p}$.
\end{proof}
\noindent Lemma \ref{lemma:muspect} indicates that whenever we take $\mu$ such that $\mu \geq \mu_{\rm S}$,  the spectral-norm distance  between the updated and initial estimators is exactly equal to the
 difference between their smallest eigenvalues.

\begin{lemma}[Scaled Frobenius norm]\label{lemma:mufrob} Assume that $\widehat{\bf \Sigma}$ is given and $\epsilon > \widehat{\gamma}_{1}$. If $\widehat{\bf \Sigma} \neq c {\bf I}$ for any scalar $c$, then
\begin{equation}\label{eqn:lemma3-1}
\left\| {\bf \Phi}_{\mu, \alpha^*}\Big(\widehat{\bf \Sigma}\Big) 
	- \widehat{\bf \Sigma} \right\|_{\rm F}
\,\geq\, (\epsilon - \widehat{\gamma}_{1}) 
 \sqrt{\frac{\sum_{i=1}^p ({\widehat{\gamma}}_i - {\overline{\widehat{\gamma}}})^2}{\sum_{i=1}^p ({\widehat{\gamma}}_i - {\widehat{\gamma}}_1)^2} } 
~~\mbox{ for all $\mu > \epsilon$},
\end{equation}
where ${\overline{\widehat{\gamma}}}$ $:=$ $\sum_i \widehat{\gamma}_i / p$.
The equality holds if and only if
\[
\mu= \mu_{\rm F} :=\frac{\sum_{i=1}^p ({\widehat{\gamma}}_i - {\widehat{\gamma}}_1)^2}{\sum_{i=1}^p ({\widehat{\gamma}}_i - {\widehat{\gamma}}_1)}.
\]
In particular,
\begin{equation}\label{eqn:mufrob1}
\left\| {\bf \Phi}_{\mu, \alpha^*}\Big(\widehat{\bf \Sigma}\Big) 
	- \widehat{\bf \Sigma} \right\|_{\rm F}
\,\leq\, \epsilon - \widehat{\gamma}_{1} ~~\mbox{ for all }\mu \geq \mu_{\rm F}.
\end{equation}
\end{lemma}
\begin{proof}%[Proof of Lemma (\ref{lemma:mufrob}).] 
We have that
\begin{eqnarray}
\frac{\epsilon - \widehat{\gamma}_{1}}{\mu - \widehat{\gamma}_{1}} \left\|   \mu {\bf I} - \widehat{\bf \Sigma}  \right\|_{\rm F}
&=&
\frac{\epsilon - \widehat{\gamma}_{1}}{\mu - \widehat{\gamma}_{1}} 
\sqrt{ p^{-1}\sum_{i=1}^p (\mu - {\widehat{\gamma}}_i)^2 } \NN \\
&=&  \frac{\epsilon - \widehat{\gamma}_{1}}{\sqrt{p}} 
\sqrt{ \frac{\sum_{i=1}^p [(\mu -{\widehat{\gamma}}_1) - ( {\widehat{\gamma}}_i - {\widehat{\gamma}}_1)]^2}{ (\mu - \widehat{\gamma}_{1})^2} } \NN \\
&=&  \frac{\epsilon - \widehat{\gamma}_{1}}{\sqrt{p}} 
\sqrt{\sum_{i=1}^p (\theta t_i - 1)^2 } \NN
\end{eqnarray}
where $\theta = \theta(\mu) = (\mu - {\widehat{\gamma}}_1)^{-1}$ and $t_i = {\widehat{\gamma}}_i - {\widehat{\gamma}}_1$.
The minimum of $h(\theta)  = \sum_{i=1}^p (\theta t_i - 1)^2$ occurs at $\theta = (\sum_{i=1}^p t_i^2)/(\sum_{i=1}^p t_i)$  (equivalently $\mu = \mu_{\rm F}$), and is equal to 
\[
\frac{p\sum_{i=1}^p (t_i - {\bar{t}})^2}{\sum_{i=1}^p t_i^2},
\] 
which proves (\ref{eqn:lemma3-1}). In addition, that $\widehat{\bf \Sigma} \neq c {\bf I}$ ensures that $\sum_{i=1}^p t_i \neq 0$ and $\sum_{i=1}^p t_i^2 \neq 0$. 
To check \eqref{eqn:mufrob1}, observe that each $(\theta t_i - 1)^2$, as a function of $\mu$, is strictly increasing on $(\mu_{\rm F}, \infty)$  and converges to 1 as $\mu \rightarrow \infty$.
\end{proof}
\noindent The assumption $\widehat{\bf \Sigma} \neq c {\bf I}$ implies that the initial estimator should not be trivial.

Finally, as a summary of Lemma \ref{lemma:muspect} and \ref{lemma:mufrob}, we have the following theorem.

\begin{theorem}[Distance between the initial and FSPD estimators]\label{thm:LSPDsumm}
Assume $\widehat{\bf \Sigma}$ and $\epsilon > 0$ are given. Set
\begin{equation}\NN
\alpha^* := \alpha^*(\mu) =
\begin{cases} 1  & \,\mbox{if }\, \widehat{\gamma}_{1} \geq \epsilon \\
	1 - \frac{\epsilon - \widehat{\gamma}_{1}}{\mu - \widehat{\gamma}_{1}} & \,\mbox{if }\,
	\widehat{\gamma}_{1} < \epsilon \end{cases} ~; \qquad
	\mu \in [\mu_{\rm SF}, \infty),
\end{equation}
where $\mu_{\rm SF} := \max \left\{ \mu_{\rm S}, \, \mu_{\rm F} \right\}$ with $\mu_{\rm S}$ and $\mu_{\rm F}$ defined in Lemma \ref{lemma:muspect} and \ref{lemma:mufrob}, respectively. Then
\BNUM
\item $\| {\bf \Phi}_{\mu, \alpha^*}\big(\widehat{\bf \Sigma}\big)
	- \widehat{\bf \Sigma} \|_2$ exactly equals to
	$(\epsilon - \widehat{\gamma}_{1})_+$
	for any $\mu$ ;
\item $\| {\bf \Phi}_{\mu, \alpha^*}\big(\widehat{\bf \Sigma}\big)
	- \widehat{\bf \Sigma} \|_{\rm F}$ is increasing
	 in $\mu$ and bounded between
\[
(\epsilon - \widehat{\gamma}_{1}) 
 \sqrt{\frac{\sum_{i=1}^p ({\widehat{\gamma}}_i - {\overline{\widehat{\gamma}}})^2}{\sum_{i=1}^p ({\widehat{\gamma}}_i - {\widehat{\gamma}}_1)^2} } 
\,\leq\,
\left\| {\bf \Phi}_{\mu, \alpha^*}\Big(\widehat{\bf \Sigma}\Big) 
	- \widehat{\bf \Sigma} \right\|_{\rm F}
\,\leq\, (\epsilon - \widehat{\gamma}_{1})_+.
\]
\ENUM
\end{theorem}
From the theorem, for any $\mu$ in $\left[ \mu_{\rm SF}, \infty \right)$,
the distance between the initial estimator $\widehat{\bf \Sigma}$ and its FSPD-updated version
is less than $(\epsilon - \widehat{\gamma}_{1})_+$.
Indeed, the Monte Carlo experiments in Section \ref{sec:simulation} show that the empirical risk of
${\bf \Phi}_{\mu, \alpha^*}\big(\widehat{\bf \Sigma}\big)$ is robust to
the choice of $\mu$. However, $\mu = \mu_{\rm SF}$ would be a preferable choice
if one wants to minimize the Frobenius-norm distance from $\widehat{\bf \Sigma}$ while maintaining the
spectral-norm distance as minimal.
We remark that a special case of the proposed FSPD estimator has been discussed
by a group of researchers (Section 5.2 in \cite{Cai2014b}).  In that study, the authors
propose to modify the initial estimator $\widehat{\bf \Sigma}$ to
$\widehat{\bf \Sigma} + (\epsilon - \widehat{\gamma}_{1}) {\bf I}$. This coincides with
our FSPD procedure with $\alpha = \alpha^*$ and $\mu \rightarrow \infty$.

\subsection{Statistical properties of the FSPD estimator}\label{subsec:convrate}

The convergence rate of
FSPD estimator to the true ${\bf \Sigma}$ is based on the triangle inequality
\begin{eqnarray}
\left\|{\bf \Phi}\Big(\widehat{\bf \Sigma}\Big) - {\bf \Sigma}\right\|
	&\leq&
	\left\|{\bf \Phi}\Big(\widehat{\bf \Sigma}\Big) - \widehat{\bf \Sigma}\right\| +
	\left\|\widehat{\bf \Sigma} - {\bf \Sigma}\right\| \NN \\
&\leq& (\epsilon - \widehat{\gamma}_{1})_+ +
	\left\|\widehat{\bf \Sigma} - {\bf \Sigma}\right\|.
\label{eqn:triangle_inequality}
\end{eqnarray}
To establish convergence rate for the modified estimator, we need to assume 

\begin{center}
\vspace{-0.4cm}
 {\bf (A1)}:  The smallest eigenvalue of $\bf \Sigma$ is such that $\epsilon < {\gamma}_{1}$.
\end{center}

\vspace{-0.4cm}
\noindent The assumption {\bf (A1)} is the same with that assumed  in \cite{Xue2012}. It  means that the cut-point $\epsilon,$ determining the PDness of the covariance matrix estimator, is set as smaller than the smallest eigenvalue of the true covariance matrix. For this choice of $\epsilon$, we claim that the convergence rate of the FSPD estimator is at least equivalent to that of the initial estimator in terms of spectral norm. 

\begin{theorem}\label{prop:LSPDcov1} 
Let $\widehat{\bf \Sigma}$ be any estimator of the true covariance matrix ${\bf \Sigma}$. Suppose 
$\epsilon$ satisfies {\bf (A1)}, $\alpha=\alpha^*$ and $\mu \in [\mu_{\rm SF}, \infty)$, 
where $\alpha^*$ and $\mu_{\rm SF}$
are defined in Theorem \ref{thm:LSPDsumm}. Then,  we have   
\[
\left\| {\bf \Phi}_{\mu, \alpha^*}\Big(\widehat{\bf \Sigma}\Big) - {\bf \Sigma} \right\|_2 \,\leq\,  
	2\left\| \widehat{\bf \Sigma} - {\bf \Sigma} \right\|_2.
\]
\end{theorem}
\begin{proof}
For $\epsilon$ satisfying {\bf (A1)}, we have
\begin{equation}\label{epsbound}
(\epsilon - \widehat{\gamma}_{1})_+ \,\leq\, ({\gamma}_{1} - \widehat{\gamma}_{1})_+
\,\leq\, \left\|\widehat{\bf \Sigma} - {\bf \Sigma}\right\|_2.
\end{equation}
 The first inequality in (\ref{epsbound})
is simply from $\epsilon \leq {\gamma}_{1}$ and  the second inequality follows from the Weyl's perturbation inequality, which states that if  ${\bf A}$ and ${\bf B}$  are general symmetric 
matrices, then 
\[
\max_i \{ \gamma_i({\bf A}) - \gamma_i({\bf B}) \} \leq \left\|{\bf A} - {\bf B}\right\|_2.
\]
Combining \eqref{eqn:triangle_inequality} and \eqref{epsbound} completes the proof.
\end{proof}

We remark that 
the arguments made for Theorem \ref{prop:LSPDcov1} are completely deterministic and there are no assumptions on how the true covariance matrix is structured or on how the estimator is defined. Thus, the convergence rate of FSPD estimator can be easily adapted to any given initial estimator.
In terms of spectral norm, Theorem \ref{prop:LSPDcov1} shows that the FSPD estimator has convergence rate at least equivalent to that of the initial estimator, which implies that  ${\bf \Phi}_{\mu, \alpha^*}\Big(\widehat{\bf \Sigma}\Big)$ is a consistent or minimax rate-optimal estimator if $\widehat{\bf \Sigma}$ is. Examples include the banding or tapering estimator \citep{Cai2010},  the adaptive block thresholding estimator \citep{Cai2012c}, the universal thresholding estimator \citep{Cai2012f},  and the adaptive thresholding estimator \citep{Cai2011b}.

To discuss the convergence rate of the FSPD estimator in Frobenius norm, we further assume
\[
{\bf (A2)}: \quad \frac{  |{\widehat{\gamma}}_1 - \gamma_1 |}{ \sqrt{\sum_{i=1}^p ({\widehat{\gamma}}_i - \gamma_i)^2 / p} } = O_p ( 1). 
\]
The assumption {\bf (A2)} implies that, in the initial estimator, the estimation error of the smallest eigenvalue
has the same asymptotic order with the averaged mean-squared error over all eigenvalues. 
With the additional assumption {\bf (A2)}, we have the following theorem.
\begin{theorem}\label{prop:LSPDcov2} 
Let $\widehat{\bf \Sigma}$ be any estimator of the true covariance matrix ${\bf \Sigma}$. Suppose 
$\alpha=\alpha^*$ and $\mu \in [\mu_{\rm SF}, \infty)$, where $\alpha^*$ and $\mu_{\rm SF}$
are defined in Theorem \ref{thm:LSPDsumm}. Then, under the assumptions {\bf (A1)} and {\bf (A2)}, 
we have 
\[
\left\| {\bf \Phi}_{\mu, \alpha^*}\Big(\widehat{\bf \Sigma}\Big) - {\bf \Sigma} \right\|_{\rm F} \,=\,  
	\Big( 1 + O_p(1) \Big) \cdot
	\left\| \widehat{\bf \Sigma} - {\bf \Sigma} \right\|_{\rm F}.
\]
\end{theorem}
\begin{proof}
Under the assumptions {\bf (A1)} and {\bf (A2)}, the inequality \eqref{eqn:triangle_inequality} shows 
\begin{eqnarray}
\left\| {\bf \Phi}_{\mu, \alpha^*}\Big(\widehat{\bf \Sigma}\Big) - {\bf \Sigma} \right\|_{\rm F} &\leq&
	(\epsilon - {\widehat{\gamma}}_1)_+  +  \left\| \widehat{\bf \Sigma} - {\bf \Sigma} \right\|_{\rm F}\nonumber\\
 	&\leq& (\gamma_1 - {\widehat{\gamma}}_1)_+  +  \left\| \widehat{\bf \Sigma} - {\bf \Sigma} \right\|_{\rm F} \nonumber \\
&=&  O_p (1) \sqrt{\sum_{i=1}^p ({\widehat{\gamma}}_i - \gamma_i)^2 / p}+  \left\| \widehat{\bf \Sigma} - {\bf \Sigma}
\right\|_{\rm F} \nonumber \\
&\le & \Big( 1 + O_p(1) \Big)    \left\| \widehat{\bf \Sigma} - {\bf \Sigma} \label{eqn:theorem3-1}
\right\|_{\rm F}
\end{eqnarray}
where the last inequality (\ref{eqn:theorem3-1}) is from Wielandt-Hoffman inequality: for any  symmetric $\bf A$ and $\bf B$,
\[
\sum_{i=1}^{p} ( \gamma_i({\bf A}) - \gamma_i({\bf B}) )^2 \,\leq\, p \left\|{\bf A} - {\bf B}\right\|_{\rm F}^2.
\]
\end{proof}

We remark that the inequality $\left\| {\bf \Phi}_{\mu, \alpha^*}\Big(\widehat{\bf \Sigma}\Big) - \widehat{{\bf \Sigma}} \right\|_{\rm F} \leq (\epsilon - {\widehat{\gamma}}_1)_+ $ can be very conservative. When $\mu$ is close to $\mu_{\rm F}$, 
Theorem~\ref{thm:LSPDsumm} implies 
\[
\left\| {\bf \Phi}_{\mu, \alpha^*}\Big(\widehat{\bf \Sigma}\Big) - \widehat{{\bf \Sigma}} \right\|_{\rm F} \approx (\epsilon - {\widehat{\gamma}}_1)_+  \sqrt{\frac{\sum_{i=1}^p ({\widehat{\gamma}}_i - {\overline{\widehat{\gamma}}})^2}{\sum_{i=1}^p ({\widehat{\gamma}}_i - {\widehat{\gamma}}_1)^2} },
\] 
where $\overline{\widehat{\gamma}}=\sum_{i=1}^p \widehat{\gamma}_i \big/p$. Given a fixed sequence of ${\widehat{\gamma}}_i$,  the term ${\sum_{i=1}^p ({\widehat{\gamma}}_i - {\overline{\widehat{\gamma}}})^2}\big/{\sum_{i=1}^p ({\widehat{\gamma}}_i - {\widehat{\gamma}}_1)^2}$ 
could be considerably smaller than one. In sequel, the Frobenius norm $\left\| {\bf \Phi}_{\mu, \alpha^*}\Big(\widehat{\bf \Sigma}\Big) - {\widehat{\bf \Sigma}} \right\|_{\rm F}$ becomes much smaller than $(\epsilon - {\widehat{\gamma}}_1)_+$, 
and possible comparable to $\left\| {\bf \Sigma} - {\widehat{\bf \Sigma}} \right\|_{\rm F}.$  Thus, the assumption
{\bf (A2)} is only a sufficient condition, and the conclusion of Theorem \ref{prop:LSPDcov2} could be true in much 
more generality.

\subsection{Computation}\label{subsec:computation}
The proposed FSPD estimator has by itself a great computational advantage over the existing estimators
that incorporate both structural regularization and PDness constraint in optimization problem
\citep{Rothman2012,Xue2012,Liu2014}. It is primarily because the FSPD procedure is optimization-free
and runs eigenvalue decomposition only once. Moreover, in practice, it can be implemented much more quickly without eigenvalue decomposition.
Recall that the FSPD estimator depends on the four functionals of $\widehat{\bf \Sigma}$: $\widehat{\gamma}_{p}$, $\widehat{\gamma}_{1}$, 
$\sum_{i=1}^p ({\widehat{\gamma}}_i - {\overline{\widehat{\gamma}}})^2$ and $\sum_{i=1}^p ({\widehat{\gamma}}_i - {\widehat{\gamma}}_1)^2$ which in turn can be written in terms of 
$\widehat{\gamma}_{p}$, $\widehat{\gamma}_{1}$, $\overline{\widehat{\gamma}}$ 
and ${\rm V}(\widehat{\gamma}) := \sum_{i=1}^p ({\widehat{\gamma}}_i - {\overline{\widehat{\gamma}}})^2$. 
For $\widehat{\gamma}_{p}$ and $\widehat{\gamma}_{1}$, the largest and
smallest eigenvalue of a symmetric matrix
can be independently calculated from the whole spectrum by the Krylov subspace method
(see Chapter 7 of \citealp{Demmel1997} or Chapter 10 of \citealp{Golub2012} for examples), which
 is faster than the usual eigenvalue decomposition of large-scale sparse matrices. Some of these
 algorithms are implemented using public and commercial software; for example, the built-in MATLAB function
 \texttt{eigs()}, which is based
 on \citealp{Lehoucq1996,Sorensen1990}, or the user-defined MATLAB function \texttt{eigifp()}
 \citep{Golub2002}, which is available online from the authors' homepage.\footnote{During the simulation, we used
 \texttt{eigifp()} instead of \texttt{eigs()}, as \texttt{eigs()} failed to converge in the
 smallest eigenvalue computations for some matrices.}
The other functionals are  written as
\begin{eqnarray}
{\overline{\widehat{\gamma}}} &=& \sum_i \widehat{\gamma}_i / p = \tr(\widehat{\bf \Sigma}) / p \NN \\
{\rm V}(\widehat{\gamma}) &=& \big(\sum_i \widehat{\gamma}_i^2 / p \big) - {\rm M}(\widehat{\gamma})^2 =
\tr(\widehat{\bf \Sigma}^2)/p - \{\tr(\widehat{\bf \Sigma}) /p \}^2, \NN
\end{eqnarray}
and can
be evaluated without the computation of the entire spectrum.

\section{Simulation study}\label{sec:simulation}

We numerically compare the finite sample performances of the
FSPD estimator and other existing estimators. For the comparison, we
use the soft thresholding estimator with universal threshold in (\ref{estimator:thr}) as
 initial regularized estimator. We show that the FSPD estimator induced by the soft thresholding
 estimator is comparable in performance to the existing PD sparse covariance matrix estimators
 and is computationally much faster.

\subsection{Empirical risk}\label{subsec:emperrcomp}

We generate samples of $p$-variate random vectors
${\bf x}_1, \ldots, {\bf x}_n$ from (i) multivariate normal distribution or (ii) multivariate $t$ distribution
 with $5$
degrees of freedom. We consider the following matrices as true ${\bf \Sigma}$:

\BNUM
\item (``Linearly tapered Toeplitz'' matrix) $[{\bf M}_1]_{ij} := \left( 1 - \frac{|i-j|}{10} \right)_+$ (when $[{\bf A}]_{ij}$ denotes
the $(i,j)$-th element of $\bf A$);
\item (``Overlapped block-diagonal'' matrix) $[{\bf M}_2]_{ij} := I(i = j) +  0.4 \, I \big( (i,j) \in (I_k \cup \{i_k + 1\}) \times (I_k \cup \{i_k + 1\}) \big)$, where the row (column) index $\{1, 2, \ldots, p\}$ is partitioned into $K := p/20$ subsets, which are non-overlapping and of equal size, and $i_k$ denotes the maximum index in $I_k$.
\ENUM
We generate 100 datasets for each of the 36 possible combinations of distribution $\in \{$multivariate
normal, multivariate-$t\}$, $\bf \Sigma \in \{{\bf M}_1, {\bf M}_2\}$, $n \in \{100,200,400\}$, and $p \in \{100,200,400\}$.

We first compute the sample covariance matrix ${\bf S} : = \frac{1}{n-1}\sum_{i=1}^{n} ({\bf x}_i - \bar{\bf x}) ({\bf x}_i - \bar{\bf x})^{\top}$,
where $\bar{\bf x} := \frac{1}{n}\sum_{i=1}^{n}{\bf x}_i$, and
the soft thresholding estimator $\widehat{\bf \Sigma}^{\rm Soft}({\lambda^*})$, where the optimal threshold
 $\lambda^*$ is chosen from the five-fold cross-validation (CV) introduced in \cite{Bickel2008a}.
 %the five-fold CV introduced in Section \ref{subsec:tuning}.
The candidate set of $\lambda$ is 
 $\big\{ \frac{k}{100} \cdot \max_{i < j} |s_{ij}| \,:\, k = 0, 1,\cdots, 100 \big\}$, a grid search from zero to the maximum of sample covariances.
 %, where $|{\bf A}|_{{\rm off, } \max}$ denotes the maximum absolute value of the elements of matrix ${\bf A}$ and ${\rm diag}({\bf S})$ is a diagonal matrix whose diagonal elements coincide with those of ${\bf S}$.

First, we investigate the spectra of the soft thresholding estimator to understand
how often and in what magnitude it violates PDness. Table \ref{table:softspectrum}
summarizes the negative eigenvalues of the soft thresholding estimates from
the simulated datasets. For both ${\bf M}_1$ and ${\bf M}_2$, we find
that the soft thresholding estimator easily lacks PDness.
In addition, the frequency of non-PDness increases as $p$ increases.

%%%%%%%%%%%%%%%%% SOFT THRESHOLDING SPECTRUM %%%%%%%%%%%%%%%%%%
%%%%%%%%%%%%%%%%%%%%%%%%%%%%%%%%%%%%%%%%%%%%%%%%%%%%%%%%%%%%%%%
%%%%%%%%%%%%%%%%%%%%%%%%%%%%%%%%%%%%%%%%%%%%%%%%%%%%%%%%%%%%%%%
\begin{table}[htb!]
{\footnotesize
\begin{center}
\begin{tabular}{cc|ccc|ccc}
\hline
& & \multicolumn{3}{c|}{${\bf M}_1$: Tapered}& \multicolumn{3}{c}{${\bf M}_2$: Overlap. block diag.} \\
$p$ & & Min. eig. & \#(Neg. eig.)/$p$ & \#(PD) & Min. eig. & \#(Neg. eig.)/p & \#(PD)  \\
\hline
 100 & $\mathcal{N}$  & -0.035 (0.004) & 0.017 (0.001) & 16/100 & 0.257 (0.003) & 0.000 (0.000) & 100/100 \\
     & $t$  & -0.066 (0.006) & 0.020 (0.001) & 9/100 & 0.292 (0.012) & 0.000 (0.000) & 100/100 \\
\hline
 200 & $\mathcal{N}$  & -0.061 (0.003) & 0.017 (0.001) & 2/100 & 0.138 (0.003) & 0.000 (0.000) & 100/100 \\
     & $t$  & -0.114 (0.020) & 0.018 (0.001) & 4/100 & 0.133 (0.018) & 0.001 (0.000) & 87/100 \\
\hline
 400 & $\mathcal{N}$  & -0.086 (0.003) & 0.016 (0.001) & 0/100 & -0.039 (0.003) & 0.006 (0.000) & 7/100 \\
     & $t$  & -0.270 (0.026) & 0.017 (0.001) & 0/100 & -0.150 (0.026) & 0.014 (0.001) & 7/100 \\
\hline
\end{tabular}
\caption{Non-PDness of the soft thresholding estimator. The columns give the minimum eigenvalue (Min. eig),
the proportion of negative eigenvalues (\#(Neg. eig.)/$p$), and the number of cases that the estimates are PD (\#(PD)) over
100 replications.}
\label{table:softspectrum}
\end{center}
}
\end{table}
%%%%%%%%%%%%%%%%%%%%%%%%%%%%%%%%%%%%%%%%%%%%%%%%%%%%%%%%%%%%%%%
%%%%%%%%%%%%%%%%%%%%%%%%%%%%%%%%%%%%%%%%%%%%%%%%%%%%%%%%%%%%%%%
In order to compare the empirical risks, we consider the following four PD covariance matrix estimators:
\BNUM
\item (``FSPD($\mu_{\rm SF}$)'') FSPD estimator ${\bf \Phi}_{\mu, \alpha^*}\big(
	\widehat{\bf \Sigma}^{\rm Soft}(\lambda^*)\big)$ with $\mu = \mu_{\rm SF}$
\item (``FSPD($\infty$)'') FSPD estimator ${\bf \Phi}_{\mu, \alpha^*}\big(
	\widehat{\bf \Sigma}^{\rm Soft}(\lambda^*)\big)$ with $\mu = \infty$
\item (``EigCon'') \cite{Xue2012}'s eigenvalue constraint estimator,
	$\widehat{\bf \Sigma}^{\rm EigCon}(\lambda^*)$ in (\ref{estimator:EigCon})
\item (``log-det'') \cite{Rothman2012}'s log-determinant barrier estimator
	$\widehat{\bf \Sigma}^{\rm logdet}(\lambda^*)$ in (\ref{estimator:logdet})
\ENUM
In applying the respective methods, we set $\epsilon = 10^{-2}$ for the two FSPD estimators and the eigenvalue constraint estimator and $\tau = 10^{-2}$ for the log-determinant barrier estimator.

Table \ref{table:emperrcomp} lists the empirical risks measured by three popular matrix norms (the matrix $l_1$, spectral,
and unscaled Frobenius norms)
for the four estimators considered. The table does not include the results for $p=$ 100, 200 of ${\bf M}_2$, because the initial soft thresholding estimates are mostly PD and empirical risks  
are almost identical across the considered estimators.

Now, we compare the empirical risks between the estimators. For the soft thresholding and FSPD approaches,
FSPD($\mu_{\rm SF}$) has lower risks (within 4\%) than the soft thresholding estimates in the matrix $l_1$ and
spectral norms and higher risks (within 4\%) in the Frobenius norm. The risks of FSPD($\infty$) are lower than
those of soft thresholding estimates in the spectral norm and higher in both the matrix $l_1$ and Frobenius norms;
however, the increase in risk does not exceed 4\%. FSPD($\mu_{\rm SF}$) demonstrates up to 2\% greater empirical
risk than FSPD($\infty$) in all the three norms for normally distributed datasets but lower empirical risk
 (within 5\%) for $t$-distributed datasets.

A comparison of the FSPD estimator and the two optimization-based estimators show that in all simulated cases, FSPD($\mu_{\rm SF}$) has
lower risks compared to both the eigenvalue constraint and log-determinant methods in the matrix $l_1$ and spectral norms,
but higher risks in the Frobenius norm; however, the difference never exceeds 4\%. In addition, FSPD($\infty$) has
a higher risk than the two optimization-based methods in both the matrix $l_1$ and Frobenius norms, except for ${\bf M}_2$
 under $p=400$, and lower risk in spectral norm. Similarly, the difference of risks are less than 4\%.

In summary, the empirical risks associated with the inspected estimators are different within
approximately 4\%; furthermore, the standard errors of these risks show that the risks of the other methods
 are within the confidence interval of each of the estimators. 
Therefore, we can conclude that the empirical
errors of the proposed FSPD estimators are comparable to the errors of the soft thresholding estimator
as well as the two optimization-based PD sparse estimators when the error is measured by the matrix $l_1$, spectral,
and Frobenius norms.

%%%%%%%%%%%%%%%%%%%%%% COV, EPS SMALL %%%%%%%%%%%%%%%%%%%%%%%
%%%%%%%%%%%%%%%%%%%%%%%%%%%%%%%%%%%%%%%%%%%%%%%%%%%%%%%%%%%%%
%%%%%%%%%%%%%%%%%%%%%%%%%%%%%%%%%%%%%%%%%%%%%%%%%%%%%%%%%%%%%
\begin{table}[htb!]
{\footnotesize
\begin{center}
\begin{tabular}{c|ccc|ccc}
\hline
 & \multicolumn{3}{c|}{ Multivariate normal } &
 	\multicolumn{3}{c}{ Multivariate $t$  }\\
\hline
 & Matrix $l_1$ & Spectral & Frobenius & Matrix $l_1$ & Spectral & Frobenius \\
\hline
\multicolumn{7}{c}{${\bf M}_1$, $p = 100$} \\
Soft thres. & 6.21 (0.11)  & 3.59 (0.05)  & 7.18 (0.07) & 9.20 (0.35)  & 5.06 (0.12)  & 10.37 (0.17) \\
 FSPD($\mu_{\rm SF}$) & 6.20 (0.11)  & 3.59 (0.05)  & 7.25 (0.07) & 9.12 (0.34)  & 5.04 (0.12)  & 10.45 (0.17) \\
 FSPD($\infty$) & 6.20 (0.11)  & 3.56 (0.05)  & 7.21 (0.07) & 9.23 (0.35)  & 5.04 (0.13)  & 10.41 (0.17) \\
 EigCon & 6.21 (0.11)  & 3.59 (0.05)  & 7.18 (0.07) & 9.19 (0.34)  & 5.06 (0.12)  & 10.37 (0.17) \\
 log-det  & 6.21 (0.11)  & 3.59 (0.05)  & 7.22 (0.06) & 9.18 (0.34)  & 5.06 (0.12)  & 10.40 (0.17) \\
\multicolumn{7}{c}{${\bf M}_1$, $p = 200$} \\
Soft thres. & 7.08 (0.08)  & 4.24 (0.04)  & 11.35 (0.06) & 13.40 (0.77)  & 6.25 (0.19)  & 17.18 (0.42) \\
 FSPD($\mu_{\rm SF}$) & 7.06 (0.08)  & 4.23 (0.04)  & 11.54 (0.05) & 13.12 (0.72)  & 6.18 (0.18)  & 17.49 (0.43) \\
 FSPD($\infty$) & 7.05 (0.08)  & 4.16 (0.04)  & 11.40 (0.05) & 13.51 (0.78)  & 6.24 (0.21)  & 17.40 (0.45) \\
 EigCon & 7.08 (0.08)  & 4.23 (0.04)  & 11.35 (0.06) & 13.29 (0.74)  & 6.25 (0.19)  & 17.18 (0.42) \\
 log-det  & 7.07 (0.08)  & 4.23 (0.04)  & 11.41 (0.05) & 13.28 (0.74)  & 6.25 (0.19)  & 17.23 (0.42) \\
\multicolumn{7}{c}{${\bf M}_1$, $p = 400$} \\
Soft thres. & 7.91 (0.08)  & 4.72 (0.03)  & 17.75 (0.06) & 19.34 (1.12)  & 7.47 (0.30)  & 25.94 (0.40) \\
 FSPD($\mu_{\rm SF}$) & 7.86 (0.07)  & 4.71 (0.03)  & 18.14 (0.06) & 18.58 (1.02)  & 7.28 (0.28)  & 26.92 (0.51) \\
 FSPD($\infty$) & 7.93 (0.08)  & 4.62 (0.03)  & 17.86 (0.06) & 19.60 (1.15)  & 7.57 (0.34)  & 26.87 (0.57) \\
 EigCon & 7.90 (0.08)  & 4.72 (0.03)  & 17.74 (0.06) & 18.92 (1.05)  & 7.45 (0.30)  & 25.92 (0.39) \\
 log-det  & 7.88 (0.08)  & 4.72 (0.03)  & 17.84 (0.06) & 18.90 (1.05)  & 7.45 (0.30)  & 25.99 (0.39) \\
\hline

%\multicolumn{7}{c}{${\bf M}_2$, $p = 100$} \\
%Soft thres. & 3.18 (0.03)  & 1.68 (0.02)  & 5.18 (0.02) & 4.95 (0.14)  & 2.30 (0.03)  & 7.56 (0.11) \\
%LSPD($\mu_{\rm SF}$) & 3.18 (0.03)  & 1.68 (0.02)  & 5.18 (0.02) & 4.95 (0.14)  & 2.30 (0.03)  & 7.56 (0.11) \\
%LSPD($\infty$) & 3.18 (0.03)  & 1.68 (0.02)  & 5.18 (0.02) & 4.95 (0.14)  & 2.30 (0.03)  & 7.56 (0.11) \\
%EigCon & 3.18 (0.03)  & 1.68 (0.02)  & 5.18 (0.02) & 4.95 (0.14)  & 2.30 (0.03)  & 7.56 (0.11) \\
%log-det  & 3.18 (0.03)  & 1.68 (0.02)  & 5.18 (0.02) & 4.95 (0.14)  & 2.30 (0.03)  & 7.56 (0.11) \\
%\multicolumn{7}{c}{${\bf M}_2$, $p = 200$} \\
%Soft thres. & 6.04 (0.05)  & 2.97 (0.03)  & 9.87 (0.03) & 9.54 (0.26)  & 4.04 (0.06)  & 14.54 (0.23) \\
%LSPD($\mu_{\rm SF}$) & 6.04 (0.05)  & 2.97 (0.03)  & 9.87 (0.03) & 9.51 (0.26)  & 4.03 (0.06)  & 14.54 (0.23) \\
%LSPD($\infty$) & 6.04 (0.05)  & 2.97 (0.03)  & 9.87 (0.03) & 9.55 (0.27)  & 4.04 (0.06)  & 14.55 (0.23) \\
%EigCon & 6.04 (0.05)  & 2.97 (0.03)  & 9.87 (0.03) & 9.52 (0.26)  & 4.04 (0.06)  & 14.54 (0.23) \\
%log-det  & 6.04 (0.05)  & 2.97 (0.03)  & 9.87 (0.03) & 9.51 (0.26)  & 4.04 (0.06)  & 14.54 (0.23) \\

\multicolumn{7}{c}{${\bf M}_2$, $p = 400$} \\
Soft thres. & 11.77 (0.10)  & 5.64 (0.05)  & 19.29 (0.07) & 19.67 (0.58)  & 7.44 (0.08)  & 27.91 (0.37) \\
 FSPD($\mu_{\rm SF}$) & 11.78 (0.10)  & 5.62 (0.05)  & 19.39 (0.07) & 19.10 (0.52)  & 7.31 (0.07)  & 28.23 (0.38) \\
 FSPD($\infty$) & 11.73 (0.10)  & 5.59 (0.05)  & 19.32 (0.07) & 19.83 (0.60)  & 7.39 (0.09)  & 28.32 (0.40) \\
 EigCon & 11.77 (0.10)  & 5.64 (0.05)  & 19.28 (0.07) & 19.34 (0.54)  & 7.44 (0.08)  & 27.89 (0.37) \\
 log-det  & 11.75 (0.10)  & 5.63 (0.05)  & 19.23 (0.07) & 19.31 (0.54)  & 7.43 (0.08)  & 27.86 (0.38) \\
\hline
\end{tabular}
\caption{Empirical risks of the PD covariance matrix esimators and soft thresholding estimator based on 100 replications. Standard errors are presented in parenthesis.}
\label{table:emperrcomp}
\end{center}
}
\end{table}
%%%%%%%%%%%%%%%%%%%%%%%%%%%%%%%%%%%%%%%%%%%%%%%%%%%%%%%%%%%%%
%%%%%%%%%%%%%%%%%%%%%%%%%%%%%%%%%%%%%%%%%%%%%%%%%%%%%%%%%%%%%
%%%%%%%%%%%%%%%%%%%%%%%%%%%%%%%%%%%%%%%%%%%%%%%%%%%%%%%%%%%%%
\subsection{Computation time}\label{subsec:timecomp}

We now numerically show the proposed FSPD estimator is much faster and simpler than optimization-based 
PD sparse estimators. We record the computation time of the four PD covariance matrix estimators  listed in Section \ref{subsec:emperrcomp} as well as the soft thresholding estimator for $n=100$ and $p=400,1200,3600$.
The distribution is multivariate normal with the true covariance matrix
${\bf \Sigma}={\bf M}_1$ and ${\bf M}_2$. Here, $\lambda^*$ is selected as in
Section \ref{subsec:emperrcomp}. The calculations in this section are performed using MATLAB running
on a computer with an Intel Core i7 CPU (3.4 GHz) and 16 GB RAM.
The two optimization-based estimators - the eigenvalue
 constraint estimator and the log-determinant barrier estimator - are solved iteratively, and the convergence criteria is set as
  $\|\widehat{\bf \Sigma}^{\rm (New)} - \widehat{\bf \Sigma}^{\rm (Old)}\| / \|\widehat{\bf\Sigma}^{\rm (Old)}\| < 10^{-7}$.

The results are summarized in Table \ref{table:timecomp}. It is not very surprising that the two FSPD estimates are
calculated extremely faster than the eigenvalue constraint and the log-determinant barrier estimates, because
 both the latter methods require iterative computations of $O(p^3)$ flops as noted in Sections \ref{subsec:optPD}
 and \ref{subsec:computation}.
In addition, FSPD($\infty$) is always faster than FSPD($\mu_{\rm SF}$), since
FSPD($\mu_{\rm SF}$) computes $\widehat{\gamma}_{1}$, $\widehat{\gamma}_{p}$, $\overline{\widehat{\gamma}}$,
 and ${\rm V}(\widehat{\gamma})$ whereas FSPD($\infty$) calculates $\widehat{\gamma}_{1}$ only.
We also note that the empirical risks from both FSPD($\mu_{\rm SF}$) and FSPD($\infty$) are
 consistent even when $n=100$ and $p=1200,3600$, which is omitted here to save space.

%%%%%%%%%%%%%%%%%%%%%% TIMING, NUMERIC %%%%%%%%%%%%%%%%%%%%%%
%%%%%%%%%%%%%%%%%%%%%%%%%%%%%%%%%%%%%%%%%%%%%%%%%%%%%%%%%%%%%
%%%%%%%%%%%%%%%%%%%%%%%%%%%%%%%%%%%%%%%%%%%%%%%%%%%%%%%%%%%%%
\begin{table}[htb!]
\begin{center}
\begin{tabular}{c|ccc|ccc}
\hline
&  \multicolumn{3}{c|}{${\bf M}_1$: Tapered}& \multicolumn{3}{c}{${\bf M}_2$: Overlap. block diag.} \\
 & $p=400$ & $p=1200$ & $p=3600$ & $p=400$ & $p=1200$ & $p=3600$ \\
 \hline
Soft thres. & 0.00  & 0.02  &0.23  &0.00 &0.03 &0.24  \\
%(NZ) & 13.6\% & 4.0\% & 1.4\% & 15.0\% & 5.7\% & 2.1\% \\
%\hline
FSPD($\mu_{\rm SF}$) & 0.01 & 0.12  & 0.66 & 0.01 & 0.13 & 0.82 \\
%(NZ) & 13.6\% & 4.0\% & 1.4\% & 15.0\% & 5.7\% & 2.1\% \\
%\hline
FSPD($\infty$) & 0.01 & 0.09  & 0.50 & 0.01 & 0.09 & 0.58 \\
%(NZ) & 13.6\% & 4.0\% & 1.4\% & 15.0\% & 5.7\% & 2.1\% \\
%\hline
EigCon  & 4.93 & 190.68  & 7757.47 & 2.42 & 106.13 & 4470.47 \\
%(NZ) & 13.5\% & 3.9\% & 1.3\% & 14.7\% & 5.5\% & 2.0\% \\
%\hline
log-det & 2.33 & 99.14  & 3157.80 & 2.30 & 97.28 & 3156.08 \\
%(NZ) & 13.1\% & 3.8\% & 1.3\% & 14.6\% & 5.5\% & 2.0\% \\
\hline
\end{tabular}
\caption{Computation time of the four PD estimators and the soft thresholding
estimator.}
% for single tuning parameter $\lambda^*$. In each cell, the first row gives the computation
%time measured in seconds and the second row indicates the proportion of the non-zero
%elements of the estimates.}
\label{table:timecomp}
\end{center}
\end{table}
%%%%%%%%%%%%%%%%%%%%%%%%%%%%%%%%%%%%%%%%%%%%%%%%%%%%%%%%%%%%%
%%%%%%%%%%%%%%%%%%%%%%%%%%%%%%%%%%%%%%%%%%%%%%%%%%%%%%%%%%%%%
%%%%%%%%%%%%%%%%%%%%%%%%%%%%%%%%%%%%%%%%%%%%%%%%%%%%%%%%%%%%%

In summary, the FSPD estimator induced by the soft thresholding estimator has empirical risk 
comparable to the two existing PD sparse covariance matrix estimators, but it is computationally 
much faster and simpler than them.

\section{Two applications}\label{sec:plug-in}

In this section, we apply the FSPD approach to two statistical procedures from the literature,
linear minimax classification and Markowitz's portfolio allocation, which require PD estimation
of covariance matrix. Both applications are illustrated with real data examples.
The linear minimax classifier is illustrated with an example from speech recognition  \citep{Tsanas2014}, and
Markowitz's portfolio allocation is illustrated with a Dow Jones stock return example \citep{Won2013}.

\subsection{Liniar minimax classifier applied to speech recognition}

\subsubsection{Linear minimax probability machine}

The \emph{linear minimax probability machine} (LMPM) proposed by \cite{Lanckriet2002} constructs a
linear binary classifier with no distributional assumptions given the mean vector and covariance matrix.
Let $({\boldsymbol \mu}_0, {\bf \Sigma}_0)$ and $({\boldsymbol \mu}_1, {\bf \Sigma}_1)$ be
a pair of the population mean vector and
 covariance matrix for group 0 and 1, respectively.
The LMPM finds a separating hyperplane that minimizes the maximum probability of misclassification over all distributions
for given ${\boldsymbol \mu}$s and ${\bf \Sigma}$s:

\begin{eqnarray}
\max_{\alpha, {\bf a} \neq 0, b} &\mbox{s.t.}&
	\inf_{ {\bf x} \sim ({\boldsymbol \mu}_0, {\bf \Sigma}_0) }
	P\left\{ {\bf a}^{\top}{\bf x} \leq b \right\} \geq \alpha \label{eqn:linearMPM} \\
&&
	\inf_{ {\bf y} \sim ({\boldsymbol \mu}_1, {\bf \Sigma}_1) }
	P\left\{ {\bf a}^{\top}{\bf x} \geq b \right\} \geq \alpha \NN
\end{eqnarray}
The solution to (\ref{eqn:linearMPM}) does not have a closed-form formula and should be numerically
computed. The PDness of the covariance matrices ${\bf \Sigma}_0$ and ${\bf \Sigma}_1$ (or their estimators)
is a necessary condition for both the convexity of the problem and the existence of
a convergent algorithm to the solution.

In practice, the true ${\boldsymbol \mu}$s and ${\boldsymbol \Sigma}$s are unknown and we must plug their
estimators into the LMPM formulation (\ref{eqn:linearMPM}).
\cite{Lanckriet2002} uses sample mean and covariance matrices when they are well defined. In case
the sample covariance matrix ${\bf S}$ is singular, they suggest using ${\bf S} + \delta {\bf I}$ with a given constant
$\delta$ rather than ${\bf S}$.

\subsubsection{Example: Speech recognition}\label{subsubsection:speech}

We illustrate the performance of the LMPM with various PD covariance matrix estimators using
voice data \citep{Tsanas2014}\footnote{The data are available on the UCI Machine Learning
Repository (\url{http://archive.ics.uci.edu/ml/}).}.
The dataset comprises 126 signals representing the pronunciation of vowel \texttt{/a/} by patients with Parkinson's disease. Each signal was pre-processed into 309 features and labeled as ``acceptable'' or ``unacceptable'' by an expert. The dataset is in the form of a $126 \times 309$ matrix with binary labels.

In order to measure classification accuracy, we randomly split the dataset 100 times into 90\% training samples and 10\% testing samples;
the LMPM is
constructed using the training samples and classification accuracy is measured using the testing samples.
In building the LMPM, the true mean vectors are
estimated from the sample means and covariance matrices are estimated from the following PD estimators:
(1) ``Sample,'' the sample covariance matrix added by $\delta {\bf I}$, where $\delta = 10^{-2}$;
 (2) ``Diagsample,'' the diagonal matrix of the sample variances;
 (3) ``LedoitWolf,'' the linear shrinkage estimator by \cite{Ledoit2004};
 (4) ``CondReg,'' the condition-number-regularized estimator by \cite{Won2013};
 (5) ``Adap.+EigCon," the eigenvalue constraint estimator by \cite{Xue2012} based on the adaptive thresholding
 estimator by \cite{Cai2011b};
 and (6) ``Adap.+FSPD,'' the proposed FSPD estimate ($\alpha=\alpha^*, \mu=\mu_{\rm SF}$) induced by the adaptive
 thresholding estimator. Here,
 unlike the numerical study in Section \ref{sec:simulation}, we use the adaptive thresholding estimator as an initial
 regularized covariance estimator
 instead of the (universal) soft thresholding estimator. This is because marginal variances of the given data can
 be unequal over variables, in which case adaptive thresholding is known to perform better than the universal
 thresholding \citep{Cai2011b}. The tuning parameter of the adaptive thresholding estimator is selected using
 five-fold CV as in \citet{Cai2011b}. The pre-determined constant $\epsilon$ for Adap.+EigCon and Adap.+FSPD is
 set as $10^{-2}$.

The average and standard deviation of classification accuracy over 100 random partitions are
 reported in Table \ref{table:linearMPM}. All the regularization methods (items 3--6) show
better classification accuracy than the naive sample covariance matrix (item 1). In addition, Adap.+FSPD has the highest
accuracy with
average 89.2\% and standard deviation 9.2\%. This record is highly competitive to the results
reported in the original paper
\citep{Tsanas2014}, which is based on a random forest and support vector machine after a feature selection algorithm named
LOGO by \citet{Sun2010}.
Finally, we note that the entire process of Adap.+FSPD, from covariance matrix estimation to LMPM construction, takes
only 0.68 seconds.

\begin{table}[h]
\begin{center}
{\small
\begin{tabular}{cccccc}
\hline
  Sample  & Diagsample & LedoitWolf & CondReg & Adap.+EigCon& Adap.+LSPD \\ \hline
  73.8 (12.4) & 76.4 (12.6)   & 86.9 (10.7) & 75.3 (13.6)  & 82.0 (18.1) & 89.2 (9.2)\\
\hline
\end{tabular}
}
\caption{The average classification accuracies (standard deviations in parenthesis) for the LMPM with selected PD covariance matrix estimators based on
100 random partitions. The abbreviations of the estimators are introduced in
the main body of the section.}
\label{table:linearMPM}
\end{center}
\end{table}

\subsection{Markowitz portfolio optimization}

\subsubsection{Minimum-variance portfolio allocation and short-sale}

In finance, portfolio refers to a family of (risky) assets held by an institution or private individual.
If there are multiple assets to invest in, a combination of assets is considered and it becomes an important issue
 to select an optimal allocation of portfolio.
\emph{Minimum variance portfolio} (MVP) optimization is one of the well-established strategies for portfolio allocation \citep{Chan1999}. The author proposes to choose a portfolio that minimizes risk, standard deviation of return. Since the MVP problem may yield an optimal allocation with short-sale or leverage, \cite{Jagannathan2003} proposes to add a no-short-sale constraint to the MVP optimization formula.

We introduce the two approaches briefly.
Let ${\bf r}:= (r_1, \ldots, r_p)^{\top}$ be a $p$-variate random vector in
which each $r_j$ represents the return of the $j$-th asset constituting portfolio ($j = 1, \ldots, p)$.
Denote by ${\bf \Sigma}:= \var ({\bf r})$ the unknown covariance matrix of assets.
A $p$-by-1 vector ${\bf w}$ represents an allocation of the investor's wealth in such a way that each $w_j$ stands for the weight of the $j$-th asset and $\sum_{j=1}^p w_j = 1$.
Then, the MVP optimization by \citet{Chan1999} is formulated as
\begin{equation}\label{eqn:minvar}
\minimize_{\bf w}\, {\bf w}^{\top}{\bf \Sigma}{\bf w} ~~\mbox{subject to}~~
	{\bf w}^{\top}{\bf 1} = 1.
\end{equation}
Note that (\ref{eqn:minvar}) allows its solution to have a weight that is negative (short-sale) or greater than 1 (leverage). 
\cite{Jagannathan2003} points out that short-sale and leveraging are sometimes impractical because of
legal constraints and considers  the MVP optimization problem with no short-sale constraint:
%These solutions with weights negative or greater than 1 are interpreted as short-sale and leverage in the stock market, respectively. Indeed, if there is a dominating factor in the true covariance structure,
%the solution is likely to have extreme
%negative weights \citep{Green1992}. However, short-sale and leveraging are sometimes impractical because of
%legal constraints. 
%
\begin{equation}\label{eqn:minvarnoshort}
\minimize_{\bf w}\, {\bf w}^{\top}{\bf \Sigma}{\bf w} ~~\mbox{subject to}~~
	{\bf w}^{\top}{\bf 1} = 1, ~ {\bf w} \geq {\bf 0},
\end{equation}
where ${\bf w} \geq {\bf 0}$ is defined component-wise.
The combination of two constraints in (\ref{eqn:minvarnoshort}) ensures that the resulting optimal weights  are restricted to $[0,1]$.
%Undoubtedly, it is problematic that the solution
%may have a bias from the true optimal weight due to nonnegativity enforcement, since a bias can inflate the
% risk of the corresponding portfolio.
%For this issue,
\cite{Jagannathan2003} empirically and analytically illustrate that (\ref{eqn:minvarnoshort})
could have a smaller risk than (\ref{eqn:minvar})
even if the no-short-sale constraint is wrong. We note that the paper handles only the case in which the sample
covariance matrix is nonsingular and plugged into the unknown ${\bf \Sigma}$ in (\ref{eqn:minvarnoshort}). In principle, ${\bf \Sigma}$ can be replaced by a suitable PD estimator.

\subsubsection{Example: Dow Jones stock return}

The aim of this analysis is not only to reproduce with different data the finding of \cite{Jagannathan2003} that
 the no-short-sale constraint does not affect the risk of (\ref{eqn:minvar}), but also to investigate empirically whether
 the same conclusion can hold for another choice of PD covariance matrix estimator.
We compare two portfolio optimization schemes: simple MVP (\ref{eqn:minvar}) and no-short-sale MVP
(\ref{eqn:minvarnoshort}). The unknown covariance matrix ${\bf \Sigma}$ is estimated with seven PD covariance
matrix estimators. Five estimators, (1) ``Sample,'' (2) ``LedoitWolf,'' (3) ``CondReg,'' (4) ``Adap.+EigCon,'' and
(5) ``Adap.+FSPD,'' are already introduced in Section \ref{subsubsection:speech}. To these, we add (6) ``POET+EigCon,''
an eigenvalue constraint estimator based on the POET estimator proposed by \cite{Fan2013} (see \citet{Xue2012}'s
discussion section in the cited paper), and (7) ``POET+FSPD,'' the FSPD estimator induced by the POET estimator.
The reason why we additionally consider (6) and (7) is that stock return data are believed to have a factor structure that
can supposedly be reflected through the POET estimation.

As data, we use the 30 stocks that constituted the Dow Jones Industrial Average in July 2008;
previously, these were used in \cite{Won2013}. The dataset contains the daily closing prices from December 1992 to June
2008, adjusting for splits and dividend distributions. The portfolios are constructed as follows.
For covariance matrix estimation, we use the stock returns of the past 60 or 240 trading days (approximately
3 or 12 months, respectively).
As Condreg, Adap, and POET require the selection of tuning parameters, they are done via five-fold CV in which the
returns of each day are treated as independent samples.
Once the portfolios are established by solving (\ref{eqn:minvar}) and (\ref{eqn:minvarnoshort}) with covariance matrices
 plugged in by the corresponding estimators, we hold each for 60 trading days. This process begins on February 18,
 1994, and is continually repeated until July 6, 2008, producing 60 holding periods in total. We record all
 the returns for each trading day and summarize them in the form of \emph{realized return},
  \emph{realized risk}, and \emph{the Sharpe ratio}. Here, the realized return and realized risk of a portfolio are defined
 as the average and standard deviation of daily returns from that portfolio, respectively. The Sharpe ratio is a risk-adjusted
performance index, defined as $\{$(realized return) - (risk-free rate)$\}$/(realized risk). The risk-free rate is set
  at 5\% per annum.

\begin{table}[h]
\begin{center}
\begin{tabular}{c|cc|cc|cc}
\hline
 & \multicolumn{2}{c|}{Realized return [\%]}  &  \multicolumn{2}{c|}{Realized risk [\%]}
 &  \multicolumn{2}{c}{Sharpe Ratio}  \\
& Simple & No Short. & Simple & No Short. & Simple & No Short. \\
\hline
\multicolumn{7}{c}{(Portfolio rebalancing based on 60 previous trading days)}\\
Sample  &  22.49 & 23.11 & 3.28 & 3.34 & 5.33 & 5.41 \\
LedoitWolf  &  21.25 & 21.61 & 3.11 & 3.06 & 5.23 & 5.44 \\
CondReg  &  24.70 & 24.70 & 4.17 & 4.16 & 4.73 & 4.73 \\
Adap.+FSPD  &  22.74 & 23.26 & 3.59 & 3.62 & 4.95 & 5.04 \\
Adap.+EigCon  &  22.74 & 23.38 & 3.35 & 3.40 & 5.30 & 5.40 \\
POET+FSPD  &  20.99 & 21.85 & 3.13 & 3.07 & 5.11 & 5.49 \\
POET+EigCon  &  20.81 & 21.28 & 3.18 & 3.06 & 4.96 & 5.32 \\
\multicolumn{7}{c}{(Portfolio rebalancing based on 240 previous trading days)}\\
Sample  &  22.42 & 23.07 & 3.33 & 3.37 & 5.24 & 5.36 \\
LedoitWolf  &  21.72 & 23.08 & 2.89 & 2.94 & 5.78 & 6.14 \\
CondReg  &  24.91 & 24.73 & 4.18 & 4.18 & 4.76 & 4.72 \\
Adap.+FSPD  &  22.68 & 23.13 & 3.57 & 3.60 & 4.95 & 5.04 \\
Adap.+EigCon  &  21.25 & 22.47 & 3.30 & 3.36 & 4.92 & 5.19 \\
POET+FSPD &  21.49 & 23.94 & 3.02 & 2.98 & 5.46 & 6.35 \\
POET+EigCon  &  20.93 & 23.49 & 3.07 & 2.99 & 5.18 & 6.19 \\
\hline
\end{tabular}
\caption{Empirical out-of-sample performances, from 30 constituents of DJIA with
 60 days of holding, starting from 2/18/1994. All the rates are annualized.}
\label{table:portfolio}
\end{center}
\end{table}

We present some interpretations and conjectures based on the results summarized in Table \ref{table:portfolio}.
We first compare the realized risks of the simple and no-short-sale MVPs; the third and fourth columns of the table
show that the differences in the realized risks of the respective MVPs are competitive.
The sample covariance matrix reproduces the findings of \cite{Jagannathan2003}, that is, that no-short-sale MVPs have
a smaller risk than simple MVPs. In addition, we find that the regularized covariance matrix estimators
(LedoitWolf, Condreg, the two Adap.s, and the two POETs) also produce similar results. Interestingly, the results of realized
returns in the first and second columns show that no-short-sale MVPs produce higher realized returns
than all simple MVPs except for Condreg, showing a relative improvement of the Sharpe ratios in all cases except
for Condreg.

Next, we compare the 60- and 240-day training results for the construction of portfolios (rows 1--6 and
7--12). Unlike the Sample and the Adap.s, both LedoitWolf and the POETs show higher realized returns and
Sharpe ratios for the 240-day training. This is true for both simple and no-short-sale MVPs. We conjecture that the
factored structure of the POET estimation captures the latent structures of stocks better in long-history data than
 short-history data. However, the higher weight on the identity matrix in LedoitWolf made the portfolio behave as
 an equal-weight investment strategy, which is practically known to work well in long-term investment.

Condreg performs somewhat differently from other regularized covariance matrix estimators; the realized return of Condreg
dominates those of the other methods. This confirms the results in \cite{Won2013}, which illustrates that an MVP with
 Condreg produces the highest wealth growth. This seems to be empirical evidence for ``high risk, high return,'' as CondReg
 shows the highest realized risk.

Finally, we compare the two PD covariance matrix estimation methods, that is, the FSPD estimators (Adap.+LSPD and
 POET+FSPD) and the eigenvalue constraint estimators (Adap.+EigCon and POET+EigCon). For POET-based estimation,
 FSPD approaches produce higher Sharpe ratios than the eigenvalue constraint methods in all cases.
 For adaptive thresholding estimation, eigenvalue constraint methods have higher Sharpe ratios than
  the FSPD estimates except for the case of a simple MVP applied for over 60 trading days. However, the POET+FSPD
 approach applied for over 240-trading days produces the highest Sharp ratios for MVPs for both the simple and
 no-short-sale cases. In addition, as the FSPD estimators are far simpler and faster than the eigenvalue constraint
 estimator, they are more desirable for practical use.

\section{Estimation of positive definite precision matrices}\label{sec:prec}

In this section, we emphasize that the FSPD approach can be applied to estimating PD precision
matrix estimators as well. Let ${\bf \Omega} = {\bf \Sigma}^{-1}$ be the unknown true precision matrix and
$\widehat{\bf \Omega}$ be one of its estimators, which possibly lacks PD. In the theory presented in Section
 \ref{sec:LSPD}, the initial estimator $\widehat{\bf \Sigma}$ can be treated as a generic symmetric matrix and
 thus Theorem \ref{thm:LSPDsumm} is valid for $\widehat{\bf \Omega}$. Therefore, we have the following theorem.

\begin{theorem}%[Convergence rate of the LSPD estimator induced by $\widehat{\bf \Omega}$]
\label{thm:LSPDprec}
Let $\widehat{\bf \Omega}$ be any estimator of the true precision matrix ${\bf \Omega}$. 
Suppose $\alpha=\alpha^*$ and $\mu \in [\mu_{\rm SF}, \infty)$, 
where $\alpha^*$ and $\mu_{\rm SF}$ are defined in Theorem \ref{thm:LSPDsumm} with 
$\widehat{\gamma}_i$, being understood by the $i$-th smallest eigenvalue of $\widehat{\bf \Omega}$.  Then, under the assumption {\bf (A1)},
we have
\[
\left\| {\bf \Phi}_{\mu, \alpha^*}\Big(\widehat{\bf \Omega}\Big) - {\bf \Omega} \right\|_2 \,\leq\,
	2\left\| \widehat{\bf \Omega} - {\bf \Omega} \right\|_2.
\]
Further, if {\bf (A2)} holds additionally,
\[
\left\| {\bf \Phi}_{\mu, \alpha^*}\Big(\widehat{\bf \Omega}\Big) - {\bf \Omega} \right\|_{\rm F} \,=\,
	\Big( 1 + O_p(1) \Big) \cdot
	\left\| \widehat{\bf \Omega} - {\bf \Omega} \right\|_{\rm F}.
\]
\end{theorem}
Theorem \ref{thm:LSPDprec} implies that, as in the estimation of the PD covariance matrix, the FSPD-updated precision
matrix estimator can preserve both the support and convergence rate.

Several regularization methods are proposed for sparse precision matrices, and, like regularized covariance matrix estimation,
 these proposed estimators do not guarantee PDness in  finite sample. Examples include CLIME
 by \citet{Cai2011}, neighborhood selection by \citet{Meinshausen2006}, SPACE by \citet{Peng2009}, symmetric lasso
 by \citet{Friedman2010}, and CONCORD by \citet{Khare2015}. In addition, penalized Gaussian likelihood methods,
 whose very definition ensures PDness of the solution, may have a solution lacking PDness, since the relevant optimization
 algorithms approximate solutions by a non-PD matrix. For instance, \cite{Mazumder2012} reports that the graphical
 lasso algorithm by \citep{Friedman2008,Witten2011} solves the problem in an indirect way and could return non-PD
 solutions.

Despite the possibility of non-PDness, we empirically find that the precision matrix estimators listed
above are likely to be PD; thus, they suffer minimally from potential non-PDness. Let $\widehat{\bf \Omega}(\lambda)$ be
 a regularized precision estimator with tuning parameter $\lambda$. Although non-PDness often arises in
 $\widehat{\bf \Omega}(\lambda)$ when $\lambda$ is not large and $\widehat{\bf \Omega}(\lambda)$ is dense,
 we observe that, with optimal selection of $\lambda^*$ through CV, $\widehat{\bf \Omega}(\lambda^*)$ is
 PD in most cases. The case ${\bf \Omega}={\bf M}_1$ (tapered Toeplitz matrix) is the only one in which we find non-PDness
 in some precision matrix estimators (CONCORD and symmetric lasso). Table \ref{table:precspectrum} summarizes the
 spectral information of these estimators.

\begin{table}[h]
{\small
\begin{center}
\begin{tabular}{cc|cc|cc}
\hline
& & \multicolumn{2}{c|}{CONCORD}& \multicolumn{2}{c}{Symmetric lasso} \\
$p$ & & Min. eig. & \#(PD) & Min. eig. & \#(PD)  \\
\hline
 100 & $\mathcal{N}$  & 2.49e-03 (2.44e-05) & 100/100 & 2.55e-03 (2.66e-05) & 100/100 \\
     & $t$  & 2.12e-03 (4.07e-05) & 100/100 & 2.23e-03 (4.89e-05) & 100/100 \\
\hline
 200 & $\mathcal{N}$  & 5.90e-04 (1.03e-05) & 100/100 & 6.55e-04 (1.03e-05) & 100/100 \\
     & $t$  & 3.05e-04 (2.41e-05) & 92/100 & 4.11e-04 (2.33e-05) & 94/100 \\
\hline
 400 & $\mathcal{N}$  & -3.44e-04 (4.09e-05) & 14/100 & -1.60e-04 (1.00e-07) & 4/100 \\
     & $t$  & -3.43e-04 (1.21e-04) & 23/100 & -2.92e-04 (2.80e-05) & 9/100 \\
\hline
\end{tabular}
\caption{Spectral information of selected sparse precision matrix estimators when the true precision matrix is
${\bf M}_1$ (tapered Toeplitz matrix).}
\end{center}
}
\label{table:precspectrum}
\end{table}

\section{Concluding remarks}\label{sec:concluding}

In this study, we propose a two-stage approach with covariance matrix regularization and the conversion toward
PDness considered to be two separate steps. Because of its two-staged nature, the proposed FSPD procedure can
 be combined with any regularized covariance or precision matrix estimator. The procedure considers a convex
 combination involving only the initial regularized estimator and the related parameters can be selected in explicit
 form. Thus, the FSPD procedure is optimization-free and can be quickly computed.
 Despite its simplicity, the FSPD estimator enjoys theoretical advantages, in that it can preserve {\it both the sparse structure and
 convergence rate} of a given initial estimator.

The FSPD procedure finds a PD matrix close to the initial estimator subject to the class
of linear shrinkage estimators. Here, we conclude the paper with a discussion on the linear shrinkage
constraint in the FSPD. Consider three classes of covariance matrices
 $\MC{S}_0 = \{ \widehat{\bf \Sigma}^* : \widehat{\bf \Sigma}^* = (\widehat{\bf \Sigma}^*)^{\top},
 \gamma_{1}(\widehat{\bf \Sigma}^*) \geq \epsilon \}$,
$\MC{S}_1 = \{ \widehat{\bf \Sigma}^* \in \MC{S}_0 : {\rm supp}(\widehat{\bf \Sigma}^*) =
{\rm supp}(\widehat{\bf \Sigma}) \}$, and
$\MC{S}_2 = \{ \widehat{\bf \Sigma}^* \in \MC{S}_0 : \widehat{\bf \Sigma}^* = \alpha
\widehat{\bf \Sigma} + (1 - \alpha) \mu {\bf I}, ~ \alpha \in [0,1], ~ \mu \in \mathbb{R} \}$.
In (\ref{eqn:ideal-opt}), we originally aim to solve
 $\min_{\widehat{\bf \Sigma}^* \in \MC{S}_1} \| \widehat{\bf \Sigma}^* -
\widehat{\bf \Sigma} \|$ given $\widehat{\bf \Sigma}$. However, in (\ref{eqn:lspd}), we solve
 $\min_{\widehat{\bf \Sigma}^* \in \MC{S}_2} \| \widehat{\bf \Sigma}^* - \widehat{\bf \Sigma} \|$ (FSPD)
 instead of (\ref{eqn:ideal-opt}), because its solution can be written explicitly. It would be of interest
 to know the cost of solving (\ref{eqn:lspd}) instead of (\ref{eqn:ideal-opt})---to be specific,
 the difference between $\min_{\widehat{\bf \Sigma}^* \in \MC{S}_2} \|
 \widehat{\bf \Sigma}^* - \widehat{\bf \Sigma} \|$ and $\min_{\widehat{\bf \Sigma}^* \in \MC{S}_1} \|
  \widehat{\bf \Sigma}^* - \widehat{\bf \Sigma} \|$.

 First, we find no cost in view of the spectral norm.
 Suppose that $\widehat{\gamma}_{1}$ and $\widehat{\gamma}_{1}^*$ are
 the minimum eigenvalues of $\widehat{\bf \Sigma}$ and $\widehat{\bf \Sigma}^*$, respectively. We
know that $\| \widehat{\bf \Sigma}^* - \widehat{\bf \Sigma} \|_2 \geq  \widehat{\gamma}_{1}^* -
 \widehat{\gamma}_{1} \geq \epsilon - \widehat{\gamma}_{1}$
for all $\widehat{\bf \Sigma}^* \in \MC{S}_0$, and thus $\min_{\widehat{\bf \Sigma}^* \in \MC{S}_0} \|
\widehat{\bf \Sigma}^* - \widehat{\bf \Sigma} \|_2 \geq \epsilon - \widehat{\gamma}_{1}$. Recall
that Lemma \ref{lemma:muspect} implies that
$\min_{\widehat{\bf \Sigma}^* \in \MC{S}_2}  \|
\widehat{\bf \Sigma}^* - \widehat{\bf \Sigma} \|_2  \leq \epsilon - \gamma_{1}$, and
\begin{equation} \nonumber
\min_{\widehat{\bf \Sigma}^* \in \MC{S}_0}  \|
\widehat{\bf \Sigma}^* - \widehat{\bf \Sigma} \|_2  \leq \min_{\widehat{\bf \Sigma}^* \in \MC{S}_1}   \|
\widehat{\bf \Sigma}^* - \widehat{\bf \Sigma} \|_2 \leq \min_{\widehat{\bf \Sigma}^* \in \MC{S}_2}  \|
\widehat{\bf \Sigma}^* - \widehat{\bf \Sigma} \|_2 \leq \epsilon - \gamma_{1}
\end{equation}
because $\MC{S}_2 \subseteq \MC{S}_1 \subseteq \MC{S}_0$. Therefore,
$\min_{\widehat{\bf \Sigma}^* \in \MC{S}_2} \|
 \widehat{\bf \Sigma}^* - \widehat{\bf \Sigma} \|_2 = \min_{\widehat{\bf \Sigma}^* \in \MC{S}_1} \| \widehat{\bf \Sigma}^*
 - \widehat{\bf \Sigma}\|_2 = \epsilon - \widehat{\gamma}_{1}$. Thus, no additional cost is given by using the linear shrinkage
 constraint in spectral norm.

 Second, in view of the Frobenius norm, we recall that  Lemma \ref{lemma:mufrob} tells
 $
 \min_{\widehat{\bf \Sigma}^* \in \MC{S}_2} \| \widehat{\bf \Sigma}^* - \widehat{\bf \Sigma} \|_{\rm F}  = (\epsilon - \widehat{\gamma}_{1}) \cdot
 \sqrt{\frac{\sum_{i=1}^p ({\widehat{\gamma}}_i - {\overline{\widehat{\gamma}}})^2}{\sum_{i=1}^p ({\widehat{\gamma}}_i - {\widehat{\gamma}}_1)^2} },
 $
 whereas, in general,
 $
 \min_{\widehat{\bf \Sigma}^* \in \MC{S}_0} \| \widehat{\bf \Sigma}^* - \widehat{\bf \Sigma} \|_{\rm F}
 = \sqrt{ \sum_{i=1}^{p} (\epsilon - \widehat{\gamma}_i)_+^2 / p }.
 $ 
The two bounds (minimums) depend on the eigenvalues of the true covariance matrix 
and the initial estimator in a complex way and hard to understand their closeness. For example, 
suppose we consider the FSPD estimator based on the soft thresholding estimator. 
A numerical investigation, not reported in this paper, shows the difference between two 
bounds decreases as $p$ increases in all the scenarios of ${\bf \Sigma}={\bf M}_2$ 
in Section \ref{sec:simulation}. However, we do not find any special pattern in the scenarios
of ${\bf \Sigma}={\bf M}_1$.


\begin{thebibliography}{}
\bibitem[Bai and Silverstein, 2010]{Bai2010}
Bai, Z. and Silverstein, J.W. (2010).
\newblock {\em Spectral Analysis of Large Dimensional Random Matrices},
\newblock Springer, New York.
\bibitem[Bickel and Levina, 2008a]{Bickel2008a}
Bickel, P. and Levina, E. (2008a).
\newblock {Covariance regularization by thresholding}.
\newblock {\em The Annals of Statistics}, {\bf 36}(6), 2577--2604.
\bibitem[Bickel and Levina, 2008b]{Bickel2008}
Bickel, P.  and Levina, E. (2008b).
\newblock {Regularized estimation of large covariance matrices}.
\newblock {\em The Annals of Statistics}, {\bf 36}(1), 199--227.
\bibitem[Bien and Tibshirani, 2011]{Bien2011}
Bien, J. and Tibshirani, R. (2011).
\newblock {Sparse estimation of a covariance matrix}.
\newblock {\em Biometrika}, {\bf 98}(4), 807--820.
\bibitem[Boyd et~al., 2010]{Boyd2010}
Boyd, S., Parikh, N., Chu, E., Peleato, B., and Eckstein, J. (2010).
\newblock {Distributed optimization and statistical learning via the
  alternating direction method of multipliers}.
\emph{Foundations and Trends in Machine Learning}, {\bf 3}(1), 1-122.
\bibitem[Cai and Liu, 2011]{Cai2011b}
Cai, T. and Liu, W. (2011).
\newblock {Adaptive thresholding for sparse covariance matrix estimation}.
\newblock {\em Journal of the American Statistical Association},
 {\bf 106}(494), 672--684.
\bibitem[Cai et~al., 2011]{Cai2011}
Cai, T., Liu, W., and Luo, X. (2011).
\newblock {A constrained $\ell_1$ minimization approach to sparse precision
  matrix estimation}.
\newblock {\em Journal of the American Statistical Association},
  {\bf 106}(494), 594--607.
\bibitem[Cai and Low, 2015]{Cai2011a}
Cai, T. and Low, M. (2015).
\newblock {A framework for estimation of convex functions}.
\emph{Statistica Sinica}, {\bf 25}, 423-456.
\bibitem[Cai et~al., 2014]{Cai2014b}
Cai, T., Ren, Z., and Zhou, H. (2014).
\newblock {Estimating structured high dimensional covariance and precision
  matrices: Optimal rates and adaptive estimation}.
\newblock Technical report available from \url{http://www-stat.wharton.upenn.edu/~tcai/}.
\bibitem[Cai and Yuan, 2012]{Cai2012c}
Cai, T. and Yuan, M. (2012).
\newblock {Adaptive covariance matrix estimation through block thresholding}.
\newblock {\em The Annals of Statistics}, {\bf 40}(4), 2014--2042.
\bibitem[Cai et~al., 2010]{Cai2010}
Cai, T., Zhang, C.H., and Zhou, H. (2010).
\newblock {Optimal rates of convergence for covariance matrix estimation}.
\newblock {\em The Annals of Statistics}, {\bf 38}(4), 2118--2144.
\bibitem[Cai and Zhou, 2012]{Cai2012f}
Cai, T.. and Zhou, H. (2012).
\newblock {Optimal rates of convergence for sparse covariance matrix
  estimation}.
\newblock {\em The Annals of Statistics}, {\bf 40}(5), 2389--2420.
\bibitem[Chan, 1999]{Chan1999}
Chan, L. (1999).
\newblock {On portfolio optimization: Forecasting covariances and choosing the risk model}.
\newblock {\em Review of Financial Studies}, {\bf 12}(5), 937--974.
\bibitem[Demmel, 1997]{Demmel1997}
Demmel, J.W. (1997).
\newblock {\em {Applied Numerical Linear Algebra}}.
\newblock SIAM, Philadelphia, PA.
\bibitem[Fan et~al., 2013]{Fan2013}
Fan, J., Liao, Y., and Mincheva, M. (2013).
\newblock {Large covariance estimation by thresholding principal orthogonal
  complements}.
\newblock {\em Journal of the Royal Statistical Society: Series B}, {\bf 75}(4), 603--680.
\bibitem[Friedman et~al., 2008]{Friedman2008}
Friedman, J., Hastie, T., and Tibshirani, R. (2008).
\newblock {Sparse inverse covariance estimation with the graphical lasso.}
\newblock {\em Biostatistics}, {\bf 9}(3), 432--441.
\bibitem[Friedman et~al., 2010]{Friedman2010}
Friedman, J., Hastie, T., and Tibshirani, R. (2010).
\newblock {Applications of the lasso and grouped lasso to the estimation of
  sparse graphical models}.
\newblock Technical report, Stanford University, Stanford, CA.
\bibitem[Golub and {Van Loan}, 2012]{Golub2012}
Golub, G.H. and {Van Loan}, C.~F. (2012).
\newblock {\em {Matrix Computations} (4th edition)}.
\newblock Johns Hopkins University Press, Baltimore.
\bibitem[Golub and Ye, 2002]{Golub2002}
Golub, G.H. and Ye, Q. (2002).
\newblock {An inverse free preconditioned Krylov subspace method for symmetric
  generalized eigenvalue problems}.
\emph{SIAM Journal on Scientific Computing}, {\bf 24}(1), 312?334.
\bibitem[Green and Hollifield, 1992]{Green1992}
Green, R.~C. and Hollifield, B. (1992).
\newblock {When will mean-variance efficient portfolios be well diversified?}
\newblock {\em The Journal of Finance}, {\bf 47}(5), 1785--1809.
\bibitem[Jagannathan and Ma, 2003]{Jagannathan2003}
Jagannathan, R. and Ma, T. (2003).
\newblock {Risk reduction in large portfolios: Why imposing the wrong
  constraints helps}.
\newblock {\em The Journal of Finance}, {\bf 58}(4),1651--1683.
\bibitem[Khare et~al., 2015]{Khare2015}
Khare, K., Oh, S.-Y., and Rajaratnam, B. (2015).
\newblock {A convex pseudolikelihood framework for high dimensional partial
  correlation estimation with convergence guarantees}.
\newblock {\em Journal of the Royal Statistical Society: Series B}, To appear.
\bibitem[Lam and Fan, 2009]{Lam2009}
Lam, C. and Fan, J. (2009).
\newblock {Sparsistency and rates of convergence in large covariance matrix
  estimation.}
\newblock {\em The Annals of Statistics}, {\bf 37}(6B), 4254--4278.
\bibitem[Lanckriet et~al., 2002]{Lanckriet2002}
Lanckriet, G., {El Ghaoui}, L., Bhattacharyya, C., and Jordan, M.I. (2002).
\newblock {A robust minimax approach to classification}.
\newblock {\em The Journal of Machine Learning Research}, {\bf 3}, 555--582.
\bibitem[Ledoit and Wolf, 2004]{Ledoit2004}
Ledoit, O. and Wolf, M. (2004).
\newblock {A well-conditioned estimator for large-dimensional covariance
  matrices}.
\newblock {\em Journal of Multivariate Analysis}, {\bf 88}(2), 365--411.
\bibitem[Lehoucq and Sorensen, 1996]{Lehoucq1996}
Lehoucq, R. and Sorensen, D. (1996).
\newblock {Deflation techniques for an implicitly restarted Arnoldi iteration}.
\newblock {\em SIAM Journal on Matrix Analysis and Applications},
  {\bf 17}(4), 789--821.
\bibitem[Liu et~al., 2014]{Liu2014}
Liu, H., Wang, L., and Zhao, T. (2014).
\newblock {Sparse covariance matrix estimation with eigenvalue constraints}.
\newblock {\em Journal of Computational and Graphical Statistics},
  {\bf 23}(2), 439--459.
\bibitem[Luenberger, 2013]{Luenberger2013}
Luenberger, D.G. (2013).
\newblock {\emph{Investment Science} (2nd edition)}.
\newblock Oxford University Press,  New York.
\bibitem[Marcenko and Pastur, 1967]{Marcenko1967}
Marcenko, V. and Pastur, L. (1967).
\newblock {Distribution of eigenvalues for some sets of random matrices}.
\newblock {\em Mathematics of the USSR-Sbornik}, {\bf 1}(4), 457--483.
\bibitem[Markowitz, 1952]{Markowitz1952}
Markowitz, H. (1952).
\newblock {Portfolio selection}.
\newblock {\em The Journal of Finance}, {\bf 7}(1), 77--91.
\bibitem[Mazumder and Hastie, 2012]{Mazumder2012}
Mazumder, R. and Hastie, T. (2012).
\newblock {The graphical lasso: New insights and alternatives}.
\newblock {\em Electronic Journal of Statistics}, {\bf 6}, 2125--2149.
\bibitem[Meinshausen and B\"{u}hlmann, 2006]{Meinshausen2006}
Meinshausen, N. and B\"{u}hlmann, P. (2006).
\newblock {High-dimensional graphs and variable selection with the Lasso}.
\newblock {\em The Annals of Statistics}, {\bf 34}(3), 1436--1462.
\bibitem[Merton, 1980]{Merton1980}
Merton, R.C. (1980).
\newblock {On estimating the expected return on the market}.
\newblock {\em Journal of Financial Economics}, {\bf 8}, 323--361.
\bibitem[Peng et~al., 2009]{Peng2009}
Peng, J., Wang, P., Zhou, N., and Zhu, J. (2009).
\newblock {Partial correlation estimation by joint sparse regression models.}
\newblock {\em Journal of the American Statistical Association},
  {\bf 104}(486), 735--746.
\bibitem[Rothman, 2012]{Rothman2012}
Rothman, A.~J. (2012).
\newblock {Positive definite estimators of large covariance matrices}.
\newblock {\em Biometrika}, {\bf 99}(3), 733--740.
\bibitem[Rothman et~al., 2009]{Rothman2009}
Rothman, A., Levina, E., and Zhu, J. (2009).
\newblock {Generalized thresholding of large covariance matrices}.
\newblock {\em Journal of the American Statistical Association},
  {\bf 104}(485), 177--186.
\bibitem[Sorensen, 1990]{Sorensen1990}
Sorensen, D. (1990).
\newblock {Implicit application of polynomial filters in a k-step Arnoldi method}.
\newblock {\em SIAM Journal on Matrix Analysis and Applications},
  13(1):357--385.
\bibitem[Stein, 1956]{Stein1956}
Stein, C. (1956).
\newblock {Inadmissibility of the usual estimator for the mean of a multivariate
normal distribution}.
\newblock In {\em Proc. Third Berkeley Symp. on Math. Statist. and Prob.}, {\bf 1}, 197--206.
\bibitem[Sun et~al., 2010]{Sun2010}
Sun, Y., Todorovic, S., and Goodison, S. (2010).
\newblock {Local learning based feature selection for high dimensional data
  analysis.}
\newblock {\em IEEE Transactions on Pattern Analysis and Machine Intelligence},
 {\bf 32}(9), 610--626.
\bibitem[Tsanas et~al., 2014]{Tsanas2014}
Tsanas, A., Little, M., Fox, C., and Ramig, L.~O. (2014).
\newblock {Objective automatic assessment of rehabilitative speech treatment in
  Parkinson's disease}.
\newblock {\em IEEE Transactions on Neural Systems and Rehabilitation
  Engineering}, 22(1):181--190.
\bibitem[Witten et~al., 2011]{Witten2011}
Witten, D., Friedman, J., and Simon, N. (2011).
\newblock {New insights and faster computations for the graphical lasso}.
\newblock {\em Journal of Computational and Graphical Statistics},
  {\bf 20}(4), 892--900.
\bibitem[Won et~al., 2013]{Won2013}
Won, J-H., Lim, J., Kim, S-J., and Rajaratnam, B. (2013).
\newblock {Condition number regularized covariance estimation}.
\newblock {\em Journal of the Royal Statistical Society. Series B}, {\bf 75}(3),
427--450.
\bibitem[Xue et~al., 2012]{Xue2012}
Xue, L., Ma, S., and Zou, H. (2012).
\newblock {Positive-definite $\ell_1$-penalized estimation of large covariance
  matrices}.
\newblock {\em Journal of the American Statistical Association},
 {\bf 107}(500),1480--1491.
\end{thebibliography}
\end{document}